\documentclass[reqno,11pt]{amsart}
\usepackage{fullpage}
\usepackage{amsfonts}
\usepackage{cite}
\usepackage{texdraw}
\usepackage[all]{xy}
\usepackage{latexsym}
\usepackage{bm}
\numberwithin{equation}{section}
\usepackage{graphics,amsmath,amssymb,amsthm,accents,epic}
\newtheorem{Thm}{Theorem}
\newtheorem{Proposition}{Proposition}
\numberwithin{equation}{section}

\def \tyb#1{\hbox{\tiny{[{\it{#1}}]}}}
\def \ty#1{\hbox{\tiny{{\it{#1}}}}}

\DeclareMathAccent{\wtilde}{\mathord}{largesymbols}{"65}
\DeclareMathAccent{\what}{\mathord}{largesymbols}{"62}

\def\m@th{\mathsurround=0pt}
\mathchardef\bracell="0365
\def\upbrall{$\m@th\bracell$}
\def\undertilde#1{\mathop{\vtop{\ialign{##\crcr
    $\hfil\displaystyle{#1}\hfil$\crcr
     \noalign
     {\kern1.5pt\nointerlineskip}
     \upbrall\crcr\noalign{\kern1pt
   }}}}\limits}

\def\underhat#1{\mathop{\vtop{\ialign{##\crcr
    $\hfil\displaystyle{#1}\hfil$\crcr
     \noalign
     {\kern1.5pt\nointerlineskip}
     \upbrall\crcr\noalign{\kern1pt
   }}}}\limits}

\def\theequation{\arabic{section}.\arabic{equation}}
\newcommand{\wb}[1]{\overline{#1}}

\newcommand{\wh}{\widehat}
\newcommand{\wt}{\widetilde}

\newcommand{\tc}{\,^t \hskip -2pt {\boldsymbol{c}}}
\newcommand{\ts}{\,^t \hskip -2pt {\boldsymbol{s}}}

\newcommand{\nn}{\nonumber}

\newcommand{\Ga}{\boldsymbol{\Gamma}}
\newcommand{\La}{\boldsymbol{\Lambda}}

\newcommand{\bF}{\boldsymbol{F}}
\newcommand{\bG}{\boldsymbol{G}}
\newcommand{\bH}{\boldsymbol{H}}
\newcommand{\bI}{\boldsymbol{I}}

\newcommand{\bK}{\boldsymbol{K}}
\newcommand{\bL}{\boldsymbol{L}}
\newcommand{\bM}{\boldsymbol{M}}

\newcommand{\bT}{\boldsymbol{T}}

\newcommand{\br}{\boldsymbol{r}}

\newcommand{\bu}{\boldsymbol{u}}
\newcommand{\bv}{\boldsymbol{v}}
\newcommand{\st}{\hbox{\tiny\it{T}}}

\newcommand{\del}{\delta}

\newcommand{\og}{\omega}
\newcommand{\kp}{\kappa}
\newcommand{\io}{\iota}

\begin{document}

\title{Revisit to solutions for Adler-Bobenko-Suris lattice equations and lattice Boussinesq-type equations}

\author{Ke Yan$^{1}$, Ying-ying Sun$^2$, Song-lin Zhao$^{1*}$\\
\\\lowercase{\scshape{${}^1$ Department of Applied Mathematics, Zhejiang University of Technology,
Hangzhou 310023, P.R. China}} \\
\lowercase{\scshape{
${}^2$ Department of Mathematics, University of Shanghai for Science and Technology, Shanghai,
200093, P.R. China}}}
\email{*Corresponding Author: songlinzhao@zjut.edu.cn}

\begin{abstract}

Solutions for all Adler-Bobenko-Suris equations excluding $\mathrm{Q4}$ and several lattice Boussinesq-type
equations are reconsidered by employing the Cauchy matrix approach. Through introducing a ``fake'' nonautonomous plane wave factor,
we derive soliton solutions, oscillatory solutions, and semi-oscillatory solutions,
for the target lattice equations. Unlike the conventional
soliton solutions, the oscillatory solutions take constant values on all elementary quadrilaterals on $\mathbb{Z}^2$,
which demonstrate periodic structure.

\end{abstract}

\keywords{Cauchy matrix approach, Adler-Bobenko-Suris lattice equations,
lattice Boussinesq-type equations, soliton solutions, (semi-)oscillatory solutions}

\maketitle

\section{Introduction}

The concept of multi-dimensional consistency, which was introduced by Nijhoff
et al. \cite{Nijhoff-MDC} and consolidated by the work of many others \cite{BS-2002,ABS-2009,Boll,LKP,ZKZ-CAC} in the early 2000's, is widely regarded as a significant breakthrough in the field of discrete integrable systems.
Among the integrable two-dimensional difference equations, a particularly noteworthy subset is comprised of equations defined on the vertices of $\mathbb{Z}^2$ lattices that are multi-dimensional consistent. The multi-dimensional consistency property for this type of quadrilateral equation
\begin{align}
\label{eq-Q}
Q(u, \wt{u}, \wh{u}, \wh{\wt{u}}; p,q)=0,
\end{align}
with a dependent variable denoted by $u:=u(n,m)$ and continuous lattice parameters $p$ and $q$,
can be interpreted geometrically as a consistency around the cube (CAC).  This property allows for an extension from a quadrilateral to a cube by adding a third dimension, such that the maps are consistent on the cube (see Figure 1). Shorthand notations $\wt{u}:=u(n+1,m),~
\wh{u}:=u(n,m+1),~\wh{\wt{u}}:=u(n+1,m+1)$ are employed in the equation \eqref{eq-Q}.  Furthermore, when a lattice equation is CAC, the equation itself, its B\"{a}cklund
transformation and its Lax pair are directly and explicitly related \cite{Atki-BT,Nij-Adler,BHQK-FCM}.
\begin{figure}[b]
\setlength{\unitlength}{0.0004in}
\hspace{2cm}
\begin{picture}(3482,2813)(0,-10)
\put(-800,1510){\makebox(0,0)[lb]{$(a)$}}
\put(1275,2708){\circle*{150}}
\put(825,2808){\makebox(0,0)[lb]{$\wh u$}}
\put(3075,2708){\circle*{150}}
\put(3375,2808){\makebox(0,0)[lb]{$\wh{\wt u}$}}
\put(1275,908){\circle*{150}}
\put(825,1008){\makebox(0,0)[lb]{$u$}}
\put(3075,908){\circle*{150}}
\put(3300,1008){\makebox(0,0)[lb]{$\wt u$}}
\drawline(275,2708)(4075,2708)
\drawline(3075,3633)(3075,0)
\drawline(275,908)(4075,908)
\drawline(1275,3633)(1275,0)
\end{picture}
\hspace{4cm}
\begin{picture}(3482,3700)(0,-500)
\put(-1200,1000){\makebox(0,0)[lb]{$(b)$}}
\put(450,1883){\circle{150}}
\put(-100,1883){\makebox(0,0)[lb]{$\wb{\wt u}$}}
\put(1275,2708){\circle*{150}}
\put(825,2708){\makebox(0,0)[lb]{$\wb u$}}
\put(3075,2708){\circle{150}}
\put(3375,2633){\makebox(0,0)[lb]{$\wb{\wh u}$}}
\put(2250,83){\circle{150}}
\put(2650,8){\makebox(0,0)[lb]{$\wh{\wt u}$}}
\put(1275,908){\circle{150}}
\put(1275,908){\circle*{90}}
\put(825,908){\makebox(0,0)[lb]{$u$}}
\put(2250,1883){\circle{150}}
\put(2250,1883){\circle{220}}
\put(2250,1883){\circle{80}}
\put(1850,2108){\makebox(0,0)[lb]{$\wb{\wh{\wt u}}$}}
\put(450,83){\circle*{150}}
\put(0,8){\makebox(0,0)[lb]{$\wt u$}}
\put(3075,908){\circle*{150}}
\put(3300,833){\makebox(0,0)[lb]{$\wh u$}}
\drawline(1275,2708)(3075,2708)
\drawline(1275,2708)(450,1883)
\drawline(450,1883)(450,83)
\drawline(3075,2708)(2250,1883)
\drawline(450,1883)(2250,1883)
\drawline(3075,2633)(3075,908)
\dashline{60.000}(1275,908)(450,83)
\dashline{60.000}(1275,908)(3075,908)
\drawline(2250,1883)(2250,83)
\drawline(450,83)(2250,83)
\drawline(3075,908)(2250,83)
\dashline{60.000}(1275,2633)(1275,908)
\end{picture}
\caption{(a): The points on which the map is defined, and (b): the consistency cube.\label{F:1}}
\end{figure}
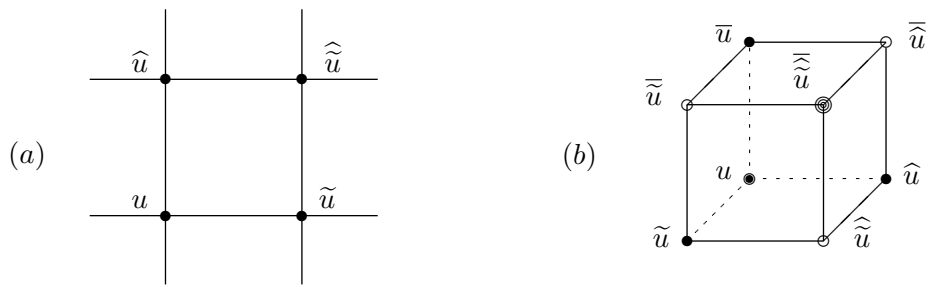

Based on the property of CAC and three additional requirements on lattice equations:
affine linear, $D_4$ symmetry and the so-called `tetrahedron property', a classification
of nonlinear integrable discrete equations was performed on the vertices
of elementary quadrangle of the $\mathbb{Z}^2$ lattice \cite{ABS-2003}. The resulting classification, known as the Adler-Bobenko-Suris (ABS) lattice list, consists of nine lattice equations: $\mathrm{H1}$, $\mathrm{H2}$, $\mathrm{H3}_{\delta}$, $\mathrm{A1}_{\delta}$,
$\mathrm{A2}$, $\mathrm{Q1}_{\delta}$, $\mathrm{Q2}$, $\mathrm{Q3}_{\delta}$, $\mathrm{Q4}$.
It should be noted that if  $\mathrm{Q4}$, also known as Adler's equation \cite{Adler-1998}, is excluded from the list, Q3$_\delta$ can act
as a top equation in the ABS list. This implies that other ``lower" equations can be obtained from the degenerations of $\mathrm{Q3}_\delta$ \cite{H-2005} (see Figure 2 in Subsection \ref{Subsec-Q3-dege}). A list of equations is compiled  in Appendix A.
Among these equations \eqref{ABS}, some are  categorized as lattice Korteweg-de Vries (KdV) type equations. For instance,  H1,
the ``lowest" member in the list, is nothing but the well-known lattice potential Korteweg-de Vries (lpKdV) equation, which firstly
appeared as a nonlinear superposition of B\"{a}cklund transformations of the potential KdV equation
\cite{WE-1973}. In addition, $\mathrm{H3}_{\delta=0}$ corresponds to the lattice
potential modified KdV (lpmKdV) equation and $\mathrm{Q1}_{\delta=0}$ is the lattice Schwarzian KdV (lSKdV)
equation, i.e., the cross-ratio equation (see the review paper \cite{lKdV}).
 It is worth noting that all these three equations can be derived  from the Nijhoff-Quispel-Capel (NQC)
equation through distinct parameter choices \cite{NQC-1983}.

In addition to the lattice maps with a single component for each lattice site, there exist three-dimensional
consistent three-component maps related to the Boussinesq (BSQ) equation (see \cite{Walker}). The lattice potential BSQ (lpBSQ) equation \cite{Tongas}, a member of lattice BSQ-type equations,  introduced as the first higher-rank case of the lattice Gel'fand-Dikii hierarchy \cite{GD}.
The bottom member of this hierarchy is the H1 equation. Similar to the lpKdV equation, the lpBSQ
equation can be obtained as nonlinear superposition formulas of the B\"{a}cklund transformations for the potential BSQ
equation \cite{BT-BSQ}. Together with the lpBSQ equation, two Miura-related three-component lattice equations, the lattice
potential modified BSQ (lpmBSQ) \cite{GD} and the lattice Schwarzian BSQ (lSBSQ) equations \cite{Schwar-1} have been also proposed.
It has been found that all these three-component lattice BSQ-type equations possess the CAC property. In \cite{H-BSQ},
a search for integrable lattice multi-component BSQ-type equations was conducted by Hietarinta, resulting in a
remarkable classification of integrable BSQ-type equations. It was subsequently proven
that all of Hietarinta's lattice BSQ-type equations arise from the so-called extended lattice BSQ systems \cite{ZZF-SAPM}.
In \cite{HZ-DBSQ}, Hietarinta and Zhang presented a comprehensive review of the lattice
BSQ-type equations, which includs continuum limits, Lax representation, Hirota bilinear forms, and soliton
solutions in terms of Casoratians,  based on the three-component forms on an elementary quadrilateral.

So far many classical methods have been applied to solve the ABS lattice list and lattice BSQ-type equations,
such as, the inverse scattering transform \cite{H1-IST,ABS-IST}, the Darboux transform \cite{H1-DT,H3-DT,mBSQ-DT},
the Hirota's bilinear method \cite{HZ-ABS,HZ-H1,HZ-BSQ,HZ-eBSQ}, and the Cauchy matrix approach
\cite{Nijhoff-ABS,ZZ-SAM-2013,FZZ-JNMP}. By virtue of the bilinear method, the fixed point idea was used
to construct seed solutions for the nonlinear lattice equations.
For the $\mathrm{H1}$ equation \eqref{H1}, there are four types of seed solutions, which are
\begin{subequations}
\label{H1-seed-solu}
\begin{align}
& -pn-qm+\mu_0, \\
& (-1)^np/2-qm+\mu_0, \\
& -pn+(-1)^mq/2+\mu_0, \\
& (-1)^np/2+(-1)^mq/2+\mu_0,
\end{align}
\end{subequations}
with  constant $\mu_0\in \mathbb{C}$. The first one in \eqref{H1-seed-solu} can be regarded as a linear background solution for $\mathrm{H1}$ equation
associated with its soliton solutions \cite{HZ-ABS} and the remaining solutions can be employed in the construction of (semi-)oscillatory solutions \cite{OS-lpKdV}. We have shown that
the lpKdV, lpmKdV and lSKdV equations possess (semi-)oscillatory solutions in \cite{FZ-OS}. The distinguishing feature among soliton solution, semi-oscillatory solution and oscillatory solution is primarily centered around  the main ingradient emerging from the respective solutions, i.e., the following discrete plane wave factor:
\begin{align*}
\begin{array}{llclcl}
\mathrm{linear~solution:}& -pn-qm+\mu_0 & \to & -pn+\frac{(-1)^m}{2}q+\mu_0 & \to & \frac{(-1)^n}{2}p+\frac{(-1)^m}{2}q+\mu_0, \\
\mathrm{plane~wave~factor:}&
\left(\frac{p+k}{p-k}\right)^n\left(\frac{q+k}{q-k}\right)^m \rho^{0}& \to &
\left(\frac{p+k}{p-k}\right)^n\frac{q-(-1)^mk}{q-k}\rho^{0} & \to &
\frac{p-(-1)^nk}{p-k}\frac{q-(-1)^mk}{q-k}\rho^{0}.
\end{array}
\end{align*}

In this paper, our aim is to take advantage of the Cauchy matrix approach developed in \cite{Nijhoff-ABS,ZZ-SAM-2013,FZZ-JNMP}  to study the
(semi-)oscillatory solutions of the ABS lattices \eqref{ABS} and some lattice BSQ-type equations.
In order to derive the (semi-)oscillatory solutions of the ABS lattices, we extend the usual plane wave factor
\begin{align}
\label{ABS-pwf-soli}
\theta(k)=\bigg(\frac{p+k}{p-k}\bigg)^n\bigg(\frac{q+k}{q-k}\bigg)^m\rho^{0},\quad \rho^{0}=\text{constant},
\end{align}
to a ``fake'' nonautonomous plane wave factor\footnote{Here we call
\eqref{ABS-pwf} a ``fake'' nonautonomous plane wave factor since it has a nonautonomous structure, while it
provides the plane wave factor for the autonomous ABS lattices.}
\begin{align}
\label{ABS-pwf}
\rho(k)=\prod^{n-1}_{i=0}\bigg(\frac{f_ip+k}{f_ip-k}\bigg)\prod^{m-1}_{j=0}\bigg(\frac{g_jq+k}{g_jq-k}\bigg)\rho^{0},
\end{align}
where $f_i:=f(i)$ and $g_j:=g(j)$ are discrete functions satisfying the conditions $f_i^2=g_j^2=1$.
When $f_i=g_j=1$, \eqref{ABS-pwf} reduces to \eqref{ABS-pwf-soli} which can be used to generate soliton solutions
of the ABS lattice equations in the framework of the Cauchy matrix method (cf. \cite{Nijhoff-ABS,ZZ-SAM-2013}). We will show that
the other choices, such as $(f_i,g_j)=((-1)^i,(-1)^j)$ or $(1,(-1)^j)$, could lead to oscillatory solutions or
semi-oscillatory solutions for the ABS lattice equations.
Similarly, for constructing the (semi-)oscillatory solutions to the lattice BSQ-type equations, we will resort to the discrete functions
$\mathcal {F}_l:=\mathcal {F}(l)$ and $\mathcal {G}_h:=\mathcal {G}(h)$
satisfying the conditions $\mathcal {F}_l^3=\mathcal {G}_h^3=1$ and  the ``fake'' nonautonomous plane wave factors
\begin{subequations}
\label{BSQ-pwf}
\begin{align}
& \varrho(k)=\prod^{n-1}_{l=0}\bigg(\frac{\mathcal {F}_lp+k}{\mathcal {F}_lp+\omega k}\bigg)
\prod^{m-1}_{h=0}\bigg(\frac{\mathcal {G}_hq+k}{\mathcal {G}_hq+\omega k}\bigg)\varrho^{0}, \\
& \sigma(k)=\prod^{n-1}_{l=0}\bigg(\frac{\mathcal {F}_lp+k}{\mathcal {F}_lp+\omega^2 k}\bigg)
\prod^{m-1}_{h=0}\bigg(\frac{\mathcal {G}_hq+k}{\mathcal {G}_hq+\omega^2 k}\bigg)\sigma^{0},
\end{align}
\end{subequations}
where parameter $\omega \neq 1$ is a cubic root of unity and $\varrho^{0},~\sigma^{0}$ are constants.

This paper is organized as follows. In Section 2, we begin by considering a determining equation set (DES), which is associated with
the ``fake'' nonautonomous plane wave factor \eqref{ABS-pwf}. From this, we introduce several
master functions $S^{(i,j)},~S(a,b)$ and $V(a)$ and use them to derive closed-form lattice KdV-type
equations. We also present soliton and (semi-)oscillatory solutions to $\mathrm{Q3}_{\delta}$,  which can be degenerated to construct solutions for all the "lower" equations in the ABS list. Section 3 focuses on soliton and (semi-)oscillatory solutions of the lattice BSQ-type equations. Section 4 provides conclusions and some remarks. Finally, three appendices are included to supplement this paper.

\section{The ABS lattice equations and solutions}
\label{sec:2}

In this section, we perform the Cauchy matrix scheme to reconsider the solutions for the ABS lattices in \eqref{ABS}.
In contrast to the conventional Cauchy matrix approach (cf. \cite{Nijhoff-ABS,ZZ-SAM-2013}),
our methodology begins with an enhanced plane wave factor \eqref{ABS-pwf}. This approach yields a minimum of three distinct solution types: solitons, oscillatory solutions and semi-oscillatory
solutions for the ABS lattice equations \eqref{ABS}.

\subsection{The Sylvester equation and master functions}

To proceed, we initially consider the following DES
\begin{subequations}
\label{ABS-DES}
\begin{align}
\label{ABS-SE}
& \bK \bM+\bM\bK=\br\, \tc, \\
\label{ABS-evo-r-p}
& (f_np\bI-\bK)\wt{\br}=(f_n p\bI+\bK)\br,  \\
\label{ABS-evo-r-q}
& (g_mq\bI-\bK)\wh{\br}=(g_m q\bI+\bK)\br,
\end{align}
\end{subequations}
in which $\bM\in \mathbb{C}_{N\times N}$ and $\br\in \mathbb{C}_{N\times 1}$
are undetermined matrices depending on independent variables $n$ and $m$, $\tc\in \mathbb{C}_{1\times N}$
and $\bK\in \mathbb{C}_{N\times N}$ are non-trivial constant matrices.  In order to ensure that the Sylvester equation is solvable \eqref{ABS-SE}, we make the assumption that $\mathcal{E}(\bK)\bigcap \mathcal{E}(-\bK)=\varnothing$, where $\mathcal{E}(\bK)$
denotes eigenvalue sets of the matrix $\bK$ (cf. \cite{Syl,Bhatia}). Note that throughout this paper, the unit matrix denoted by $\bI$ will be utilized, with the index indicating its size being omitted.

Now we introduce some master functions
\begin{subequations}
\label{ABS-MF}
\begin{align}
\label{ABS-Sij}
& S^{(i,j)}=\tc\,\bK^j(\bI+\bM)^{-1}\bK^i\br, \quad i,j\in \mathbb{Z}, \\
\label{ABS-Sab}
& S(a,b)=\tc(b\bI+\bK)^{-1}(\bI+\bM)^{-1}(a\bI+\bK)^{-1}\br, \quad a,b\in \mathbb{C}, \\
\label{ABS-Va}
& V(a)=1-\tc(\bI+\bM)^{-1}(a\bI+\bK)^{-1}\br,
\end{align}
\end{subequations}
which play a crucial role in the construction of closed-form lattice equations.

In terms of the Sylvester equation \eqref{ABS-SE}, we have the following symmetric properties (see \cite{ZZ-SAM-2013}), i.e.,
\begin{subequations}
\label{ABS-Symm}
\begin{align}
\label{ABS-Sijab-Symm}
& S^{(i,j)}=S^{(j,i)}, \quad S(a,b)=S(b,a), \\
& V(a)=1-\tc(\bI+\bM)^{-1}(a\bI+\bK)^{-1}\br=1-\tc(a\bI+\bK)^{-1}(\bI+\bM)^{-1}\br,
\end{align}
\end{subequations}
which remain invariant under similarity transformations. Indeed, suppose that matrix $\bar{\bK}$ is similar to $\bK$
using a transformation matrix $\bT$, i.e.,
\begin{subequations}
\label{trans-sim}
\begin{align}
\label{K1-K}
\bar{\bK}=\bT \bK \bT^{-1}.
\end{align}
We denote
\begin{align}
\label{Mrs-1}
\bar{\bM}=\bT \bM \bT^{-1},~~\bar{\br}=\bT \br,~~\bar{\tc}=\tc \bT^{-1}.
\end{align}
\end{subequations}
and then we easily get
\begin{subequations}
\begin{align}
& S^{(i,j)}=\tc\bK^j(\bI+\bM)^{-1}\bK^i\br
= \bar{\tc} \bar{\bK}^j(\bI+\bar{\bM})^{-1}\bar{\bK}^i\bar{\br}, \\
& S(a,b)=\tc(b\bI+\bK)^{-1}(\bI+\bM)^{-1}(a\bI+\bK)^{-1}\br \nn\\
& \qquad\quad  = \bar{\tc}(b\bI+\bar{\bK})^{-1}(\bI+\bar{\bM})^{-1}(a\bI+\bar{\bK})^{-1}\bar{\br}, \\
& V(a)=1-\tc(\bI+\bM)^{-1}(a\bI+\bK)^{-1}\br=1-\bar{\tc}(\bI+\bar{\bM})^{-1}(a\bI+\bar{\bK})^{-1}\bar{\br},
\end{align}
\end{subequations}
thus establishing the similarity invariance of these master functions.

By the DES \eqref{ABS-DES} we can derive shift relations for the fundamental functions
\eqref{ABS-MF}. As for the shift relations associated with the master function $S^{(i,j)}$,
the results are presented in the following proposition.
\begin{Proposition}
\label{ABS-Prop-Sij-sf}
For the master function $S^{(i,j)}$ defined by \eqref{ABS-Sij}, provided that the matrices $\bM, \bK$ and vectors $\br, \tc$ satisfy the DES \eqref{ABS-DES},
the following relations
\begin{subequations}
\label{ABS-Sij-sf}
\begin{align}
\label{ABS-Sij-sf-p}
& f_n p\wt{S}^{(i,j)}-\wt{S}^{(i,j+1)}=f_n pS^{(i,j)}+S^{(i+1,j)}-S^{(0,j)}\wt{S}^{(i,0)},\\
\label{ABS-Sij-sf-q}
& g_m q\wh{S}^{(i,j)}-\wh{S}^{(i,j+1)}=g_m qS^{(i,j)}+S^{(i+1,j)}-S^{(0,j)}\wh{S}^{(i,0)},
\end{align}
\end{subequations}
pertaining the shift operators $\wt{\phantom{a}}$ and $\wh{\phantom{a}}$ are valid.
\end{Proposition}

\begin{proof}

Here we just demonstrate the first shift relation, since the second one
is a similar relation with the replacements of $p$ by $q$, $\wt{\phantom{a}}$  by $\wh{\phantom{a}}$,
and $f_n$ by $g_m$.

To begin, we consider the shift relation of $\bM$.
Subtracting \eqref{ABS-SE} from \eqref{ABS-SE}$\wt{\phantom{a}}$, and using \eqref{ABS-SE} and \eqref{ABS-evo-r-p}, we have
\begin{align}
\label{ABS-M-sf-1}
(f_np\bI-\bK)\wt{\bM}=(f_np\bI+\bK)\bM.
\end{align}
Substituting the Sylvester equation \eqref{ABS-SE} into \eqref{ABS-M-sf-1} to replace
$\bK\bM$, we arrive at
\begin{align}
\label{ABS-M-sf-2}
(f_np\bI-\bK)\wt{\bM}-\bM(f_np\bI-\bK)=\br\tc.
\end{align}

To construct the shift relation \eqref{ABS-Sij-sf-p}, we introduce an auxiliary vector function
\begin{align}
\label{ABS-ui}
\bu^{(i)}=(\bI+\bM)^{-1}\bK^i\br, \quad i\in \mathbb{Z},
\end{align}
which relates to $S^{(i,j)}$ by
\begin{align}
\label{ABS-Sij-ui}
S^{(i,j)}=\tc\,\bK^j\bu^{(i)}.
\end{align}

Multiplying both sides of \eqref{ABS-ui} from the left by the
matrix $(\bI+\bM)$ and taking $\wt{\phantom{a}}$-shift, then the
utilization of \eqref{ABS-M-sf-2} leads to
\begin{align}
\label{ui-sh-a}
(f_np\bI-\bK)\wt{\bu}^{(i)}=f_np\bu^{(i)}+\bu^{(i+1)}-\bu^{(0)}\wt{S}^{(i,0)},
\end{align}
which moreover yields \eqref{ABS-Sij-sf-p} by left multiplying $\tc\bK^j$
and using \eqref{ABS-Sij-ui}.
\end{proof}

Noting that the symmetric property \eqref{ABS-Sijab-Symm}, one can easily deduce
the other two shift relations
\begin{subequations}
\label{ABS-Sij-sf-34}
\begin{align}
& f_n p\wt{S}^{(i,j)}-\wt{S}^{(i+1,j)}=f_n pS^{(i,j)}+S^{(i,j+1)}-S^{(i,0)}\wt{S}^{(0,j)}, \\
& g_m q\wh{S}^{(i,j)}-\wh{S}^{(i+1,j)}=g_m qS^{(i,j)}+S^{(i,j+1)}-S^{(i,0)}\wh{S}^{(0,j)}.
\end{align}
\end{subequations}

The following proposition reveals the shift relations of the master functions $S(a,b)$ and $V(a)$.
\begin{Proposition}
\label{ABS-Prop-SabVa-sf}
For the master functions $S(a,b)$ and $V(a)$ defined by \eqref{ABS-Sab} and \eqref{ABS-Va}, provided that the matrices $\bM, \bK$ and vectors $\br, \tc$ satisfy the DES \eqref{ABS-DES},
the following relations
\begin{subequations}
\label{ABS-SabVa-sf-1}
\begin{align}
\label{ABS-SabVa-1a}
& 1-(f_np+b)\wt{S}(a,b)+(f_np-a)S(a,b)=\wt{V}(a)V(b), \\
\label{ABS-SabVa-1b}
& 1-(g_mq+b)\wh{S}(a,b)+(g_mq-a)S(a,b)=\wh{V}(a)V(b).
\end{align}
\end{subequations}
pertaining the shift operators $\wt{\phantom{a}}$ and $\wh{\phantom{a}}$ are valid.
\end{Proposition}

\begin{proof}

We introduce an auxiliary vector function
\begin{align}
\label{ABS-ua}
\bu(a)=(\bI+\bM)^{-1}(a\bI+\bK)^{-1}\br,
\end{align}
namely
\begin{align}
\label{ABS-ua-Def}
(\bI+\bM)\bu(a)=(a\bI+\bK)^{-1}\br.
\end{align}
Multiplying \eqref{ABS-ua-Def}$\wt{\phantom{a}}$ from the left by a factor $f_n p\bI-\bK$
and using \eqref{ABS-M-sf-2}, then we have
\begin{align}
\label{ABS-ua-sf}
(f_np\bI-\bK)\wt{\bu}(a)=(f_np-a)\bu(a)+\bu^{(0)}\wt{V}(a),
\end{align}
which yields \eqref{ABS-SabVa-1a} after left-multiplying \eqref{ABS-ua-sf} by
$\tc(b\bI+\bK)^{-1}$ and using formulas $S(a,b)=\tc(b\bI+\bK)^{-1}\bu(a)$
and $V(b)=1-\tc(b\bI+\bK)^{-1}\bu^{(0)}$. The relation \eqref{ABS-SabVa-1b} can be derived
from \eqref{ABS-SabVa-1a} by replacing $(p,\wt{\phantom{a}},f_n)$ by $(q,\wh{\phantom{a}},g_m)$.
Thus we finish the verification.
\end{proof}

In a similar manner, noticing the symmetric property \eqref{ABS-Symm}, we also have
\begin{subequations}
\label{ABS-SabVa-sf-2}
\begin{align}
\label{ABS-SabVa-2a}
& 1-(f_np+a)\wt{S}(a,b)+(f_np-b)S(a,b)=V(a)\wt{V}(b), \\
\label{ABS-SabVa-2b}
& 1-(g_mq+a)\wh{S}(a,b)+(g_mq-b)S(a,b)=V(a)\wh{V}(b).
\end{align}
\end{subequations}

\subsection{The lattice KdV-type equations and NQC equation}
\label{sec:3}

We now introduce some variables
\begin{align}
\label{ABS-vari}
w=S^{(0,0)}, \quad v=S^{(-1,0)}-1, \quad z=S^{(-1,-1)}
-\sum_{i=0}^{n-1}(f_ip)^{-1}-\sum_{j=0}^{m-1}(g_jq)^{-1}+z_0,
\end{align}
where $z_0\in \mathbb{C}$. By a similar discussion as that in \cite{Nijhoff-ABS}, some lattice KdV-type equations
can be obtained from shift relations \eqref{ABS-Sij-sf}, \eqref{ABS-Sij-sf-34},
\eqref{ABS-SabVa-sf-1} and \eqref{ABS-SabVa-sf-2}, expressed in a closed-form and presented below.

\noindent{\it 1. lpKdV equation:}
\begin{subequations}
\begin{align}
\label{lpKdV-w}
(f_np-g_mq+\wh{w}-\wt{w})(f_np+g_mq+w-\wh{\wt{w}})=p^2-q^2,
\end{align}
which is rewritten as
\begin{align}
\label{lpKdV-mu}
(\wh{\mu}-\wt{\mu})(\mu-\wh{\wt{\mu}})=p^2-q^2,
\end{align}
\end{subequations}
under the point transformation
\begin{align}
\label{Tr-mu-u}
\mu=w-\sum_{i=0}^{n-1}f_ip-\sum_{j=0}^{m-1}g_jq+\mu_0, \quad \mu_0\in\mathbb{C}.
\end{align}

\noindent{\it 2. lpmKdV equation:}
\begin{subequations}
\label{lpmKdV-vnu}
\begin{align}
\label{eq:v}
\wh{\wt{v}}(f_np\wt{v}-g_mq\wh{v})=v(f_np\wh{v}-g_mq\wt{v}),
\end{align}
which is rewritten as
\begin{align}
\label{lpmKdV}
\nu(p\wh{\nu}-q\wt{\nu})=\wh{\wt{\nu}}(p\wt{\nu}-q\wh{\nu}),
\end{align}
\end{subequations}
under the point transformation
\begin{align}
\label{Tr-nu-v}
\nu=\prod^{n-1}_{i=0}f_i\prod^{m-1}_{j=0}g^{-1}_jv.
\end{align}

\noindent{\it 3. lSKdV equation:}
\begin{align}
\label{lSKdV}
p^2(z-\wt{z})(\wh{z}-\wh{\wt{z}})=q^2(z-\wh{z})(\wt{z}-\wh{\wt{z}}).
\end{align}

\noindent{\it 4. ``nonautonomous'' NQC equation:}
\begin{align}
\label{NQC}
\frac{1-(f_np+b)\wh{\wt{S}}(a,b)+(f_np-a)\wh{S}(a,b)}{1-(g_mq+b)\wh{\wt{S}}(a,b)+(g_mq-a)\wt{S}(a,b)}
=\frac{1-(g_mq+a)\wh{S}(a,b)+(g_mq-b)S(a,b)}{1-(f_np+a)\wt{S}(a,b)+(f_np-b)S(a,b)}.
\end{align}

\noindent{\bf Remark 1.} \textit{In the above lattice KdV-type equations, equations \eqref{lpKdV-w}
and \eqref{eq:v} are autonomous, since they can be transformed into equations \eqref{lpKdV-mu}
and \eqref{lpmKdV} through transformations \eqref{Tr-mu-u} and \eqref{Tr-nu-v}, respectively. The
lSKdV equation \eqref{lSKdV} itself is autonomous. However, for the NQC equation \eqref{NQC},
it is autonomous when $f_n=g_m=1$, while it is nonautonomous when $f_n=(-1)^n$ or $g_m=(-1)^m$
since it can not be transformed to the autonomous version.}

\noindent{\bf Remark 2.} \textit{According to
different values of $f_i$ and $g_j$, the term $-\sum\limits_{i=0}^{n-1}f_ip-\sum\limits_{j=0}^{m-1}g_jq+\mu_0$ in \eqref{Tr-mu-u}
becomes diverse seed solutions \eqref{H1-seed-solu} for the equation \eqref{lpKdV-mu}, i.e.,
\begin{align*}
\mu=
\begin{cases}
\text{linear seed solution}: -np-mq+\mu_0,~ & \text{if} \quad f_i=g_j=1, \\
\text{oscillatory seed solution}: (-1)^np/2+(-1)^mq/2+\mu_0,~ & \text{if} \quad f_i=(-1)^i,~g_j=(-1)^j, \\
\text{semi-oscillatory seed solution}: -np+(-1)^mq/2+\mu_0,~ & \text{if} \quad f_i=1,~g_j=(-1)^j, \\
\text{semi-oscillatory seed solution}: (-1)^np/2-mq+\mu_0,~ & \text{if} \quad f_i=(-1)^i,~g_j=1.
\end{cases}
\end{align*}
}

\subsection{$\mathrm{Q3}_{\delta}$ and degeneration}
\label{Subsec-Q3-dege}

It is well-established that the autonomous NQC equation yields a 4-to-1 relationship solution for
both $\mathrm{Q3}_{\delta=0}$ and $\mathrm{Q3}_{\delta}$ depending on the sign choices of two additional parameters, $a$ and $b$ \cite{Nijhoff-ABS}.
Despite the NQC equation \eqref{NQC} is nonautonomous when $f_n=(-1)^n$ or $g_m=(-1)^m$,
we present evidence that this equation still provides a four-term solution to $\mathrm{Q3}_{\delta}$.

\begin{Thm}
\label{Thm-Q3-solu}
The function
\begin{align}
\label{eq:Q3sol}
\begin{split}
u= &  A\digamma(a,b)\left[ 1-(a+b)S(a,b)\right]+B\digamma(a,-b)\left[ 1-(a-b)S(a,-b)\right]\\
& + C\digamma(-a,b)\left[ 1+(a-b)S(-a,b)\right]+ D\digamma(-a,-b)\left[ 1+(a+b)S(-a,-b)\right]
\end{split}
\end{align}
solves $\mathrm{Q3}_{\delta}$ \eqref{Q3}, where $S(\pm a,\pm b)$ are the solutions of
the NQC equation \eqref{NQC} with parameters $\pm a$, $\pm b$.
The function $\digamma(a,b)$ is defined by
\begin{align}
\label{eq:vpdef}
\digamma(a,b) =\prod^{n-1}_{i=0}\bigg(\frac{P}{(f_ip-a)(f_ip-b)}\bigg)\prod^{m-1}_{j=0}\bigg(\frac{Q}{(g_jq-a)(g_jq-b)}\bigg),
\end{align}
and $P,Q$ are given by \eqref{ell-Q3}; and $A$, $B$, $C$ and $D$ are constants subject to a single constraint:
\begin{align}\label{eq:ABCD}
AD(a+b)^2-BC(a-b)^2=-\delta^2/(16ab).
\end{align}
\end{Thm}
The proof of Theorem \ref{Thm-Q3-solu} can be demonstrated in a manner analogous to that presented in \cite{Nijhoff-ABS}. Consequently, the proof is omitted herein.

Similarly to the rational soliton case  described in \cite{Nijhoff-ABS},  it is possible to derive solutions for the ``lower''
lattice equations $\mathrm{Q2}$, $\mathrm{Q1}_{\delta}$, $\mathrm{H3}_{\delta}$,
$\mathrm{H2}$ and $\mathrm{H1}$ equations from the solution for $\mathrm{Q3}_{\delta}$ by adjusting the parameters $a$ and $b$ and appropriately modifying the coefficients $A$, $B$, $C$, and $D$. This process follows the scheme depicted in Figure 2.
\vspace{-5mm}
\begin{center}
\begin{displaymath}
\xymatrix{ \boxed{\mathrm{Q3_{\del}}} \ar[d] \ar[r] & \boxed{\mathrm{Q2}} \ar[d] \ar[r] & \boxed{\mathrm{Q1_{\del}}} \ar[d] \\
\boxed{\mathrm{H3_{\del}}} \ar[r] & \boxed{\mathrm{H2}} \ar[r] & \boxed{\mathrm{H1}} }
\end{displaymath}
\begin{minipage}{11cm}{\footnotesize\qquad\qquad\qquad\qquad\qquad
{FIGURE. 2} Degeneration relations}
\end{minipage}
\end{center}
The upper horizontal sequence in this scheme, involving the degenerations
of the $\mathrm{Q}$ equations, is obtained from performing careful limits of the type $b\rightarrow a$,
while the vertical limit from $\mathrm{Q}$ to $\mathrm{H}$ equations is obtained from the limits $a$ or $b\rightarrow 0$.

\vspace{.3cm}
\noindent{$\mathrm{Q3}_{\delta}\rightarrow \mathrm{Q2}$:} Inserting the degeneration
\begin{align}
b=a(1-2\epsilon), \quad u\rightarrow \delta\left(1/\epsilon+1+(1+2u)\epsilon\right)/(4a^2)
\end{align}
into $\mathrm{Q3}_{\delta}$ leads to the $\mathrm{Q2}$. Meanwhile, the four constants $A, B, C$ and $D$, constrained by
\eqref{eq:ABCD} are replaced by three new constants $A$, $D$ and $\xi_0$ as
\begin{align}
& A \rightarrow \delta A\epsilon/(4a^2), \quad
B \rightarrow \delta\left(1/\epsilon+1-\xi_0+((3+\xi_0^2)/2+2AD)\epsilon\right)/(8a^2), \\
& C \rightarrow \delta\left(1/\epsilon+1+\xi_0+((3+\xi_0^2)/2+2AD)\epsilon\right)/(8a^2), \quad
D \rightarrow \delta D\epsilon/(4a^2).
\end{align}
Then the solution to $\mathrm{Q2}$ reads
\begin{align}
\label{eq:Q2sol}
u =& ((\xi+\xi_0)^2+1)/4+a(\xi+\xi_0)S(-a,a)+a^2\left(Z(a,-a)+Z(-a,a)\right)+AD \\
& +\big(A\rho(a)(1-2aS(a,a))+D\rho(-a)(1+2aS(-a,-a))\big)/2,
\end{align}
where
\begin{subequations}
\begin{align}
\label{eq:xidef}
& \xi=\xi_{n,m}=2a\left(\sum^{n-1}_{i=0}(f_ip)/(a^2-p^2)+\sum^{m-1}_{j=0}(g_jq)/(a^2-q^2)\right), \\
& Z(a,-a)=-\partial_bS(a,b)|_{b=-a}, \quad Z(-a,a)=Z(a,-a)|_{a\rightarrow -a}, \\
& \rho(a)=\prod^{n-1}_{i=0}\bigg(\frac{f_ip+a}{f_ip-a}\bigg)\prod^{m-1}_{j=0}\bigg(\frac{g_jq+a}{g_jq-a}\bigg).
\end{align}
\end{subequations}

\vspace{.3cm}
\noindent{$\mathrm{Q2}\rightarrow \mathrm{Q1}_{\delta}$:} By taking the degeneration
\begin{align}
u\rightarrow \delta^2/(4\epsilon^2)+u/\epsilon,
\end{align}
one can deduce the $\mathrm{Q1}_{\delta}$ from $\mathrm{Q2}$. Replacing constants $A$, $D$ and $\xi_0$ by
\begin{align}
A \rightarrow 2A/\epsilon, \quad
D \rightarrow 2D/\epsilon, \quad
\xi_0 \rightarrow \xi_0+2B/\epsilon,
\end{align}
then we find the solutions for $\mathrm{Q1}_{\delta}$
\begin{align}
u = A\rho(a)(1-2aS(a,a))+B(\xi+\xi_0+2aS(-a,a))+D\rho(-a)(1+2aS(-a,-a)),
\end{align}
where the constants $A, B, D$ and $\xi_0$ satisfy the constraint
\begin{align}
AD+B^2/4=\delta^2/16.
\end{align}

\vspace{.3cm}
\noindent{$\mathrm{Q3}_{\delta}\rightarrow \mathrm{H3}_{\delta}$:} Implementing
\begin{align}
b=1/\epsilon^2, \quad u\rightarrow \sqrt{\delta}\epsilon^3 u/2
\end{align}
and
\begin{align}
A\rightarrow \epsilon^3\sqrt{\delta}A/2, \quad
B\rightarrow \epsilon^3\sqrt{\delta}B/2, \quad
C\rightarrow \epsilon^3\sqrt{\delta}C/2, \quad
D\rightarrow \epsilon^3\sqrt{\delta}D/2,
\end{align}
the solution to $\mathrm{H3}_{\delta}$ is formulated as
\begin{align}
u=(A+(-1)^{n+m}B)\vartheta V(a)+((-1)^{n+m}C+D)\vartheta^{-1}V(-a),
\end{align}
in which function $\vartheta$ is
\begin{align}
\vartheta=\prod^{n-1}_{i=0}\left(\frac{P}{a-f_ip}\right)\prod^{m-1}_{j=0}\left(\frac{Q}{a-g_jq}\right),
\end{align}
where $P,Q$ are defined by \eqref{eq:parcurves} and the constants $A$, $B$, $C$ and $D$ are subject to a single constraint
\begin{align}
AD-BC=-\delta/(4a).
\end{align}

\vspace{.3cm}
\noindent{$\mathrm{Q2}\rightarrow \mathrm{H2}$:} To derive the solution proposed by $\mathrm{H2}$ from that formulated by $\mathrm{Q2}$, we consider
\begin{align}
a=1/\epsilon, \quad u\rightarrow 1/4+\epsilon^2 u.
\end{align}
The degenerations of the constants $A, D$ and $\xi_0$ are
\begin{align}
A\rightarrow  A(\epsilon+\zeta_1\epsilon^2/2), \quad D\rightarrow A(-\epsilon+\zeta_1\epsilon^2/2), \quad
\xi_0\rightarrow \epsilon \zeta_0,
\end{align}
where $A$, $\zeta_0$ and $\zeta_1$ are unconstraint constants. The H2 solution is described as
\begin{align}
u =(\zeta+\zeta_0)^2/4 -(\zeta+\zeta_0) S^{(0,0)} + 2S^{(0,1)}-A^2+(-1)^{n+m}A(\zeta+\zeta_1-2S^{(0,0)}),
\end{align}
where $\zeta=2\big(\sum\limits^{n-1}_{i=0}f_ip+\sum\limits^{m-1}_{j=0}g_jq\big)$.

\vspace{.3cm}
\noindent{$\mathrm{Q1}_{\delta}\rightarrow \mathrm{H1}$:} The degeneration from $\mathrm{Q1}_{\delta}$ solution to
$\mathrm{H1}$ solution can be obtained by setting
\begin{align}
a=1/\epsilon, \quad u\rightarrow \epsilon \delta u,
\end{align}
as well as
\begin{align}
A\rightarrow \delta A(1+\zeta_1\epsilon)/2, \quad D\rightarrow \delta A(-1+\zeta_1\epsilon)/2,  \quad
B\rightarrow \delta B, \quad \xi_0 \rightarrow \epsilon \zeta_0.
\end{align}
The resulting solution of $\mathrm{H1}$ is of the form
\begin{align}
u=B(\zeta+\zeta_0-2S^{(0,0)})+(-1)^{n+m}A(\zeta+\zeta_1-2S^{(0,0)}),
\end{align}
where $\zeta_0$, $\zeta_1$, $A$ and $B$ subject to a single constraint
\begin{align}
A^2-B^2=-1/4.
\end{align}

\subsection{Exact solutions}

According to the analysis of the above subsections, we recognize that solutions to the
ABS list \eqref{ABS} are given by the master functions $S^{(i,j)},S(a,b)$ and $V(a)$, where
$\tc, \br, \bM$ and $\bK$ are defined by the DES \eqref{ABS-DES}.
Therefore, to derive exact solutions for these equations, we just need to solve the DES \eqref{ABS-DES}. Because of
the similarity invariance of these master functions and the covariance of the DES \eqref{ABS-DES} under
transformations \eqref{trans-sim}, here we turn to solve \eqref{ABS-DES} with $\bK$ being its Jordan canonical form, i.e.
\begin{subequations}
\label{ABS-DES-JC}
\begin{align}
\label{ABS-SE-JC}
& \Ga \bM+\bM\Ga=\br\, \tc, \\
\label{ABS-r-JC-a}
& (f_np\bI-\Ga)\wt{\br}=(f_np\bI+\Ga)\br,  \\
\label{ABS-r-JC-b}
& (g_mq\bI-\Ga)\wh{\br}=(g_mq\bI+\Ga)\br,
\end{align}
\end{subequations}
where $\Ga$ is the Jordan canonical form of the matrix $\bK$, satisfying $\mathcal{E}(\Ga)\bigcap \mathcal{E}(-\Ga)=\varnothing$.

Equations \eqref{ABS-r-JC-a} and \eqref{ABS-r-JC-b} are linear and imply
\begin{align}
\br=(f_np\bI+\Ga)(f_np\bI-\Ga)^{-1}(g_mq\bI+\Ga)(g_mq\bI-\Ga)^{-1}\br^0,
\end{align}
where $\br^0$ is a $N$-th constant column vector. Since the Sylvester equation
\eqref{ABS-SE-JC} was solved by factorizing $\bM$ into triplet $\bF\bG\bH$ in \cite{ZZ-SAM-2013},
here we just list the most general mixed solutions
for $\br$ and $\bM$ (A set of notations is given in the Appendix B).
\begin{Thm}
\label{T:3}
For the equation set \eqref{ABS-DES-JC} with generic
\begin{align}
\label{Ga-gen-T}
\Ga=\mathrm{Diag}\bigl(\Ga^{\tyb{N$_1$}}_{\ty{D}}(\{k_j\}^{N_1}_{1}),
\Ga^{\tyb{N$_2$}}_{\ty{J}}(k_{N_1+1}),\Ga^{\tyb{N$_3$}}_{\ty{J}}(k_{N_1+2}),\cdots,
\Ga^{\tyb{N$_s$}}_{\ty{J}}(k_{N_1+(s-1)})\bigr)
\end{align}
and
\begin{align}
\label{ct-T}
\tc=(c_1,c_2,\cdots,c_{N_1},c_{N_1+1},\cdots, c_{N_1+N_2+\cdots+N_s}),
\end{align}
we have solutions
\begin{subequations}
\label{r-M-g}
\begin{align}
\br=\left(
\begin{array}{c}
\br_{\ty{D}}^{\tyb{N$_1$}}(\{k_j\}_{1}^{N_1})\\
\br_{\ty{J}}^{\tyb{N$_2$}}(k_{N_1+1})\\
\br_{\ty{J}}^{\tyb{N$_3$}}(k_{N_1+2})\\
\vdots\\
\br_{\ty{J}}^{\tyb{N$_s$}}(k_{N_1+(s-1)})
\end{array}
\right),~~~
\bM=\bF\bG \bH,
\end{align}
where
\begin{align}
&\bF=\mathrm{Diag}\bigl(
\bF^{\tyb{N$_1$}}_{\ty{D}}(\{k_j\}^{N_1}_{1}),
\bF^{\tyb{N$_2$}}_{\ty{J}}(k_{N_1+1}),\bF^{\tyb{N$_3$}}_{\ty{J}}(k_{N_1+2}),\cdots,
\bF^{\tyb{N$_s$}}_{\ty{J}}(k_{N_1+(s-1)})
\bigr),\label{r-M-g-F}\\
&\bH=\mathrm{Diag}\bigl(
\bH^{\tyb{N$_1$}}_{\ty{D}}(\{c_j\}^{N_1}_{1}),
\bH^{\tyb{N$_2$}}_{\ty{J}}(\{c_j\}^{N_1+N_2}_{N_1+1}),
\cdots,
\bH^{\tyb{N$_s$}}_{\ty{J}}(\{c_j\}^{N_1+N_2+\cdots+N_s}_{N_1+N_2+\cdots+N_{s-1}+1})\bigr),\label{r-M-g-H}
\end{align}
and $\bG$ is a symmetric matrix with block structure
\begin{align}
\label{r-M-g-G1}
\bG=\bG^T=(\bG_{i,j})_{s\times s}
\end{align}
with
\begin{align}
\begin{array}{ll}
\bG_{1,1}=\bG^{\tyb{N$_1$}}_{\ty{D}}(\{k_j\}^{N_1}_{1}),&~\\
\bG_{1,j}=\bG_{j,1}^T=\bG^{\tyb{N$_1$,N$_j$}}_{\ty{DJ}}(\{k_j\}^{N_1}_{1};k_{N_{j-1}+1}),&~~(1<j\leq s), \\
\bG_{i,j}=\bG_{j,i}^T=\bG^{\tyb{N$_i$,N$_j$}}_{\ty{JJ}}(k_{N_{i-1}+1},k_{N_{j-1}+1}),&~~(1<i\leq j\leq s).
\end{array}
\label{r-M-g-G2}
\end{align}
\end{subequations}
Besides, in addition to $\Ga$, $\tc$, $\br$ and $\bM$ mentioned above, the pair
\begin{align}\label{ArM}
\{\mathcal{A}\br, \, \mathcal{A}\bM\}
\end{align}
with the same $\Ga$ and $\tc$ is also a solution to the equation set \eqref{ABS-DES-JC}. $\mathcal{A}$ in \eqref{ArM} is of the form
\[\mathcal{A}=\mathrm{Diag}(\bI,\mathcal{A}_2,\mathcal{A}_3,\cdots,\mathcal{A}_s),\]
in which $\mathcal{A}_j$ is
a $N_j$-th order constant lower triangular Toeplitz matrix (cf. \cite{ZDJ-Wron,ZZSZ}).
\end{Thm}

\noindent{\bf Remark 3.} \textit{When $f_n=g_m=1$, we obtain the usual soliton-Jordan mixed solutions, which have been
reported in \cite{ZZ-SAM-2013}. While when $f_n=(-1)^n$ or $g_m=(-1)^{m}$, then we can get the (semi-)oscillatory solutions for the
ABS lattice list \eqref{ABS}.}

Next several examples of solutions for the lpKdV (or $\mathrm{H1}$) equation \eqref{lpKdV-mu} are listed,
with the notation
\begin{align}
\rho_{ij}=\rho_{i}c_{j}, \quad e^{A_{12}}=\bigg(\frac{k_1-k_2}{k_1+k_2}\bigg)^2, \quad
\eta_1=\sum^{n-1}_{i=0}\frac{2f_ip}{p^2-k_1^{2}}+\sum^{m-1}_{j=0}\frac{2g_jq}{q^2-k_1^{2}}.
\end{align}
In the case of $\Ga=k_1$, we have
\begin{align}
\label{ABS-1SS}
\mu=\frac{2k_1\rho_{11}}{2k_1+\rho_{11}}-\sum_{i=0}^{n-1}f_ip-\sum_{j=0}^{m-1}g_jq+\mu_0.
\end{align}
In the case of $\Ga=\mathrm{Diag}(k_1,k_2)$, the corresponding solution reads
\begin{align}
\label{ABS-2SS}
\mu=\frac{4k_1k_2(\rho_{11}+\rho_{22})+2(k_1+k_2)e^{A_{12}}\rho_{11}\rho_{22}}{4k_1k_2+2k_2\rho_{11}+2k_1\rho_{22}+e^{A_{12}}\rho_{11}\rho_{22}}
-\sum_{i=0}^{n-1}f_ip-\sum_{j=0}^{m-1}g_jq+\mu_0.
\end{align}
In the case of $\Ga=\left(\begin{array}{cc}
k_1 & 0   \\
1  & k_1
\end{array}\right)$, we get the simplest Jordan-block solution
\begin{align}
\label{ABS-JBS}
\mu=\frac{16(\rho_{11}+k_1^{4}\eta_1\rho_{12})-4k_1\rho_{12}^{2}}{16k_1^{4}+8k_1^{3}(\rho_{11}+\eta_1\rho_{12})-8k_1^{2}\rho_{12}-\rho_{12}^{2}}
-\sum_{i=0}^{n-1}f_ip-\sum_{j=0}^{m-1}g_jq+\mu_0.
\end{align}

The solutions \eqref{ABS-1SS}-\eqref{ABS-JBS} have a linear or (semi-)oscillatory
background part $-\sum\limits_{i=0}^{n-1}f_ip-\sum\limits_{j=0}^{m-1}g_jq+\mu_0$. In order to show the figures of these solutions
we ignore this part and only illustrate the first part $w$. Figure 3 exhibits the solution
$w$ with a linear seed solution $-pn-qm+\mu_0$, where we have taken $f_i=g_j=1$ in the background part.

\vskip20pt
\begin{center}
\begin{picture}(120,100)
\put(-180,-23){\resizebox{!}{3.5cm}{\includegraphics{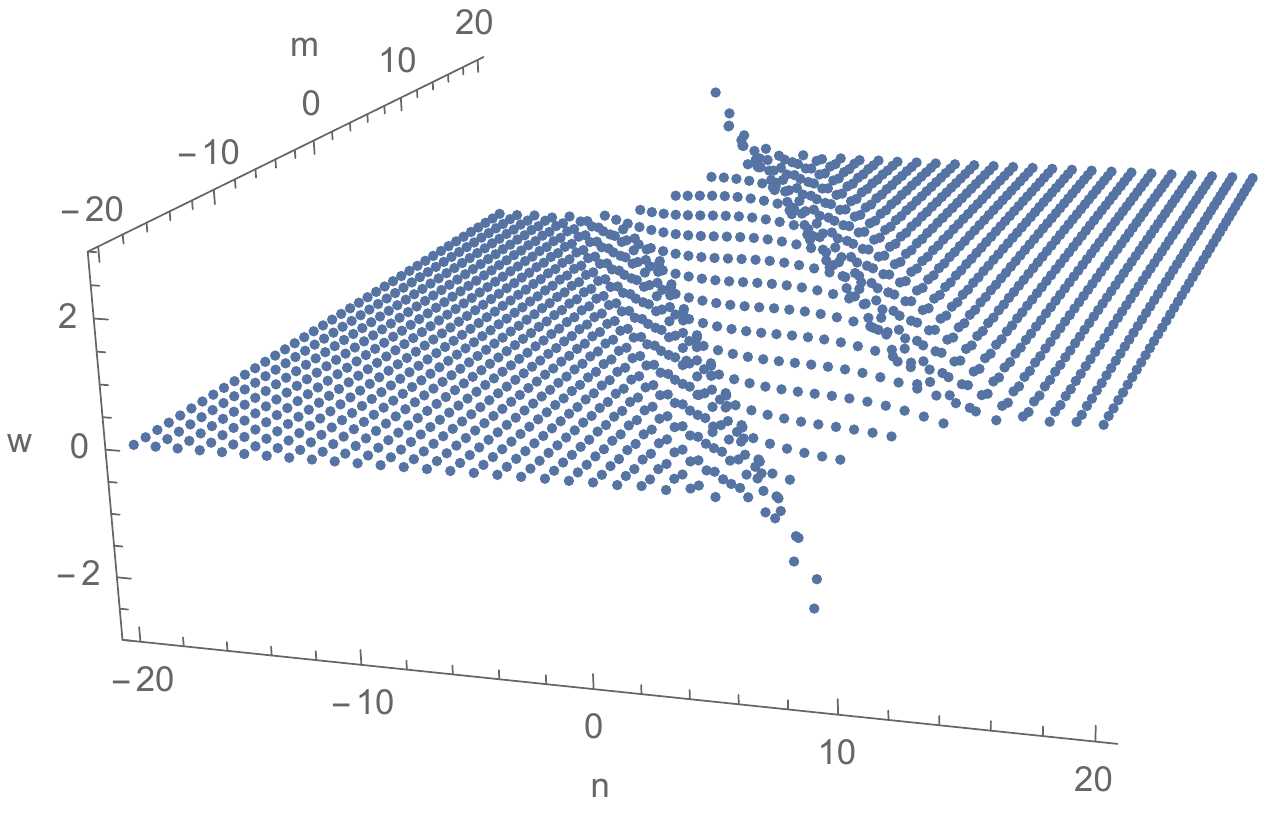}}}
\put(-10,-23){\resizebox{!}{3.5cm}{\includegraphics{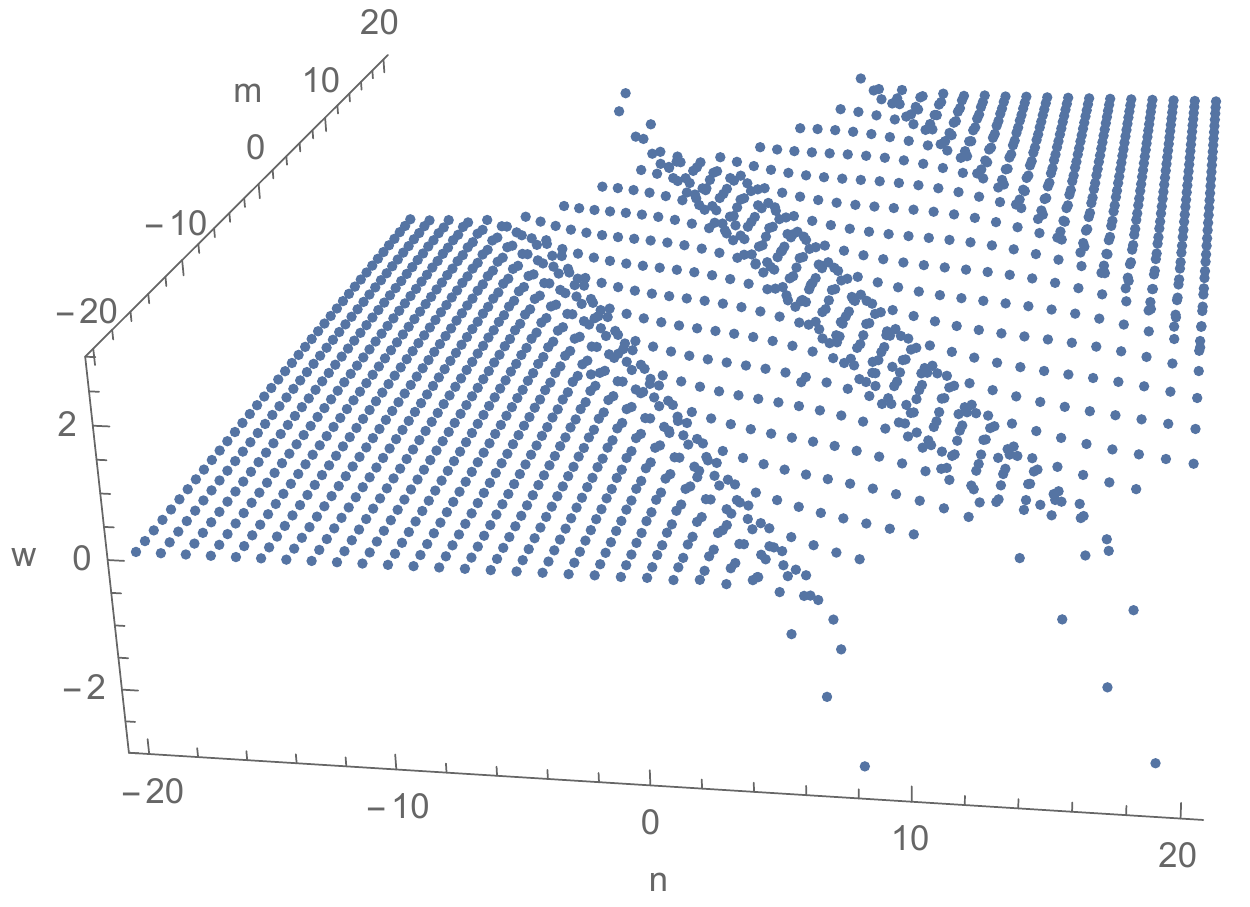}}}
\put(150,-23){\resizebox{!}{3.5cm}{\includegraphics{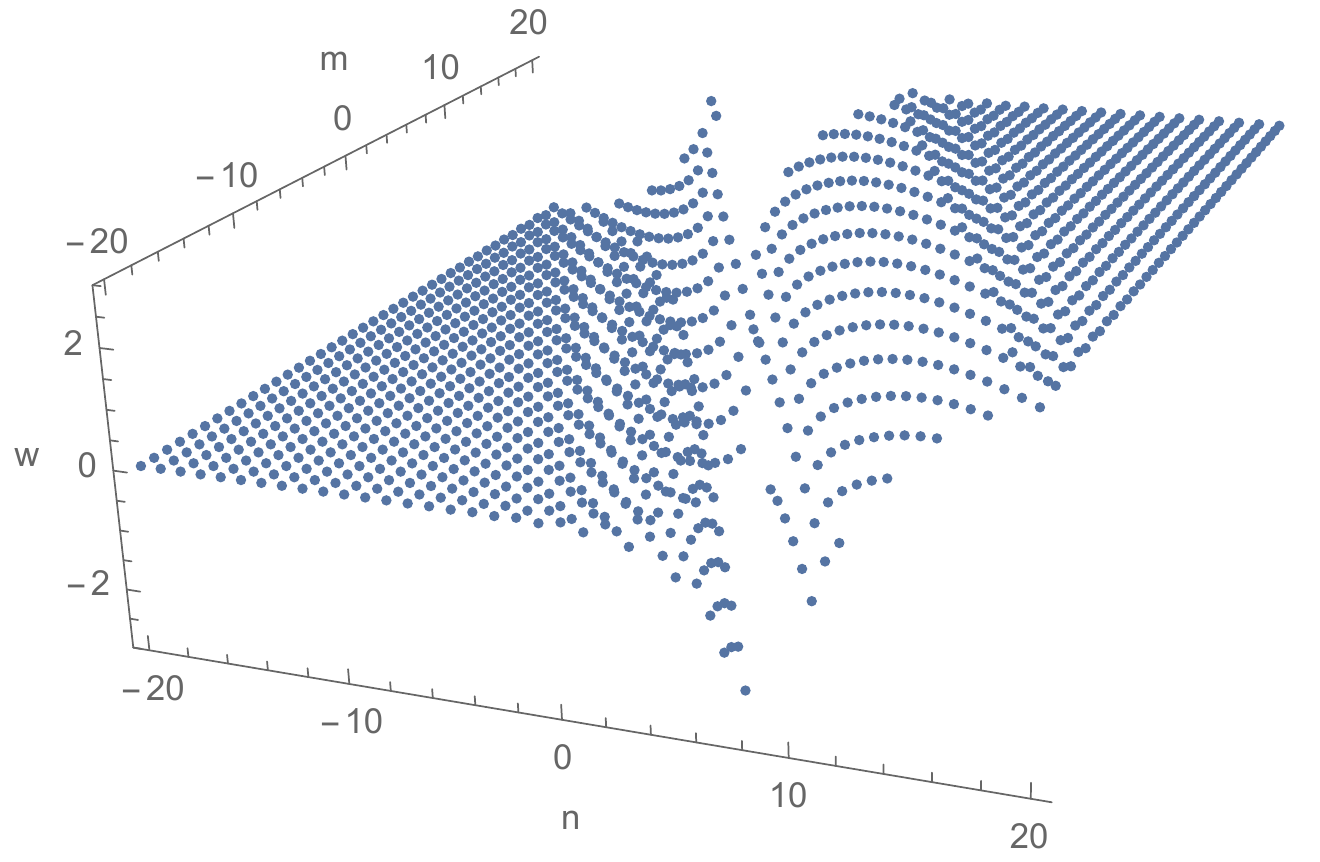}}}
\end{picture}
\end{center}
\vskip 20pt
\begin{center}
\begin{minipage}{15cm}{\footnotesize
\qquad\qquad\quad(a)\qquad\qquad\qquad\qquad\qquad\qquad\qquad\qquad (b) \qquad\qquad\qquad\qquad\qquad\qquad\qquad\quad\quad (c)\\
{FIGURE 3.} Shape and motion of solutions $w$ given by \eqref{ABS-1SS}-\eqref{ABS-JBS} with $p=0.1, q=0.2$ and $\rho_i^0c_j=1$:
(a) One-soliton solution \eqref{ABS-1SS} with $k_1=0.7$; (b) Two-soliton solution \eqref{ABS-2SS} with $k_1=0.7$ and $k_2=0.6$;
(c) Jordan-block solution \eqref{ABS-JBS} with $k_1=0.7$.}
\end{minipage}
\end{center}

For $(f_i,g_j)=((-1)^i,(-1)^j)$, $w$ (also $\mu$) takes four constant values on all elementary quadrilaterals on $\mathbb{Z}^2$
in terms of the parity of $n$ and $m$. For instance,
$w$ given by \eqref{ABS-1SS} can be expressed in Table 1.
\begin{center}
\footnotesize \setlength{\tabcolsep}{8pt}
\renewcommand{\arraystretch}{1.5}
\begin{tabular}[htbp]{|l|l|l|l|}
\hline
$(n,m)$
& Solution $w$
\\
\hline
(even,even)
& $w=\frac{2k_1\rho^0_{1}c_1}{2k_1+\rho^0_{1}c_1}$ \\
\hline
(odd,even)
& $w=\frac{2k_1(p+k_1)\rho^0_{1}c_1}{2k_1(p-k_1)+(p+k_1)\rho^0_{1}c_1}$\\
\hline
(even,odd)
& $w=\frac{2k_1(q+k_1)\rho^0_{1}c_1}{2k_1(q-k_1)+(q+k_1)\rho^0_{1}c_1} $\\
\hline
(odd,odd)
& $w=\frac{2k_1(p+k_1)(q+k_1)\rho^0_{1}c_1}{2k_1(p-k_1)(q-k_1)+(p+k_1)(q+k_1)\rho^0_{1}c_1}$ \\
\hline
\end{tabular}
\end{center}
\begin{center}
\begin{minipage}{9cm}{\footnotesize
{TABLE 1. \emph{$w$ given by \eqref{ABS-1SS} as $(f_i,g_j)=((-1)^i,(-1)^j)$.}}}
\end{minipage}
\end{center}

It is worth to mention that for a given variable $w$, the other three values in Table 1 are appropriately associated with $\wt{w}$, $\wh{w}$ and $\wh{\wt{w}}$.
In addition, $w$ possesses periodic property
for the discrete independent variable $n$ or $m$ with minimal positive period $2$. Solution for H1 in this case is sketched in Figure 4.
These properties also hold for the solutions to lpmKdV \eqref{lpmKdV-vnu}, lSKdV \eqref{lSKdV}, nonautonomous NQC \eqref{NQC} and
other equations in the ABS list \eqref{ABS}. Figure 5 depicts the behavior of semi-oscillatory solution given by \eqref{ABS-1SS}
with $(f_i,g_j)=(1,(-1)^j)$.

\vskip20pt
\begin{center}
\begin{picture}(120,100)
\put(-180,-23){\resizebox{!}{3.5cm}{\includegraphics{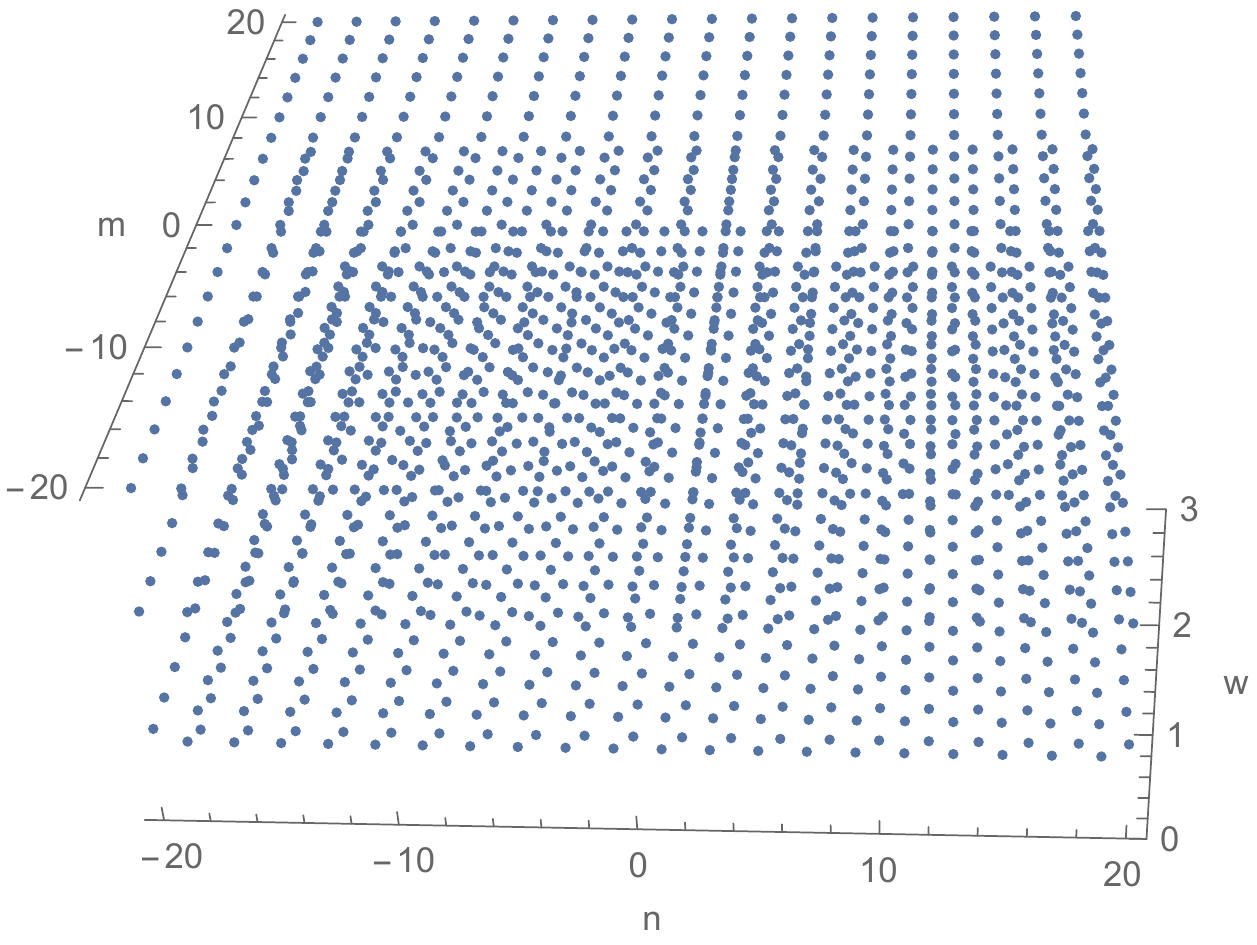}}}
\put(-20,-23){\resizebox{!}{3.5cm}{\includegraphics{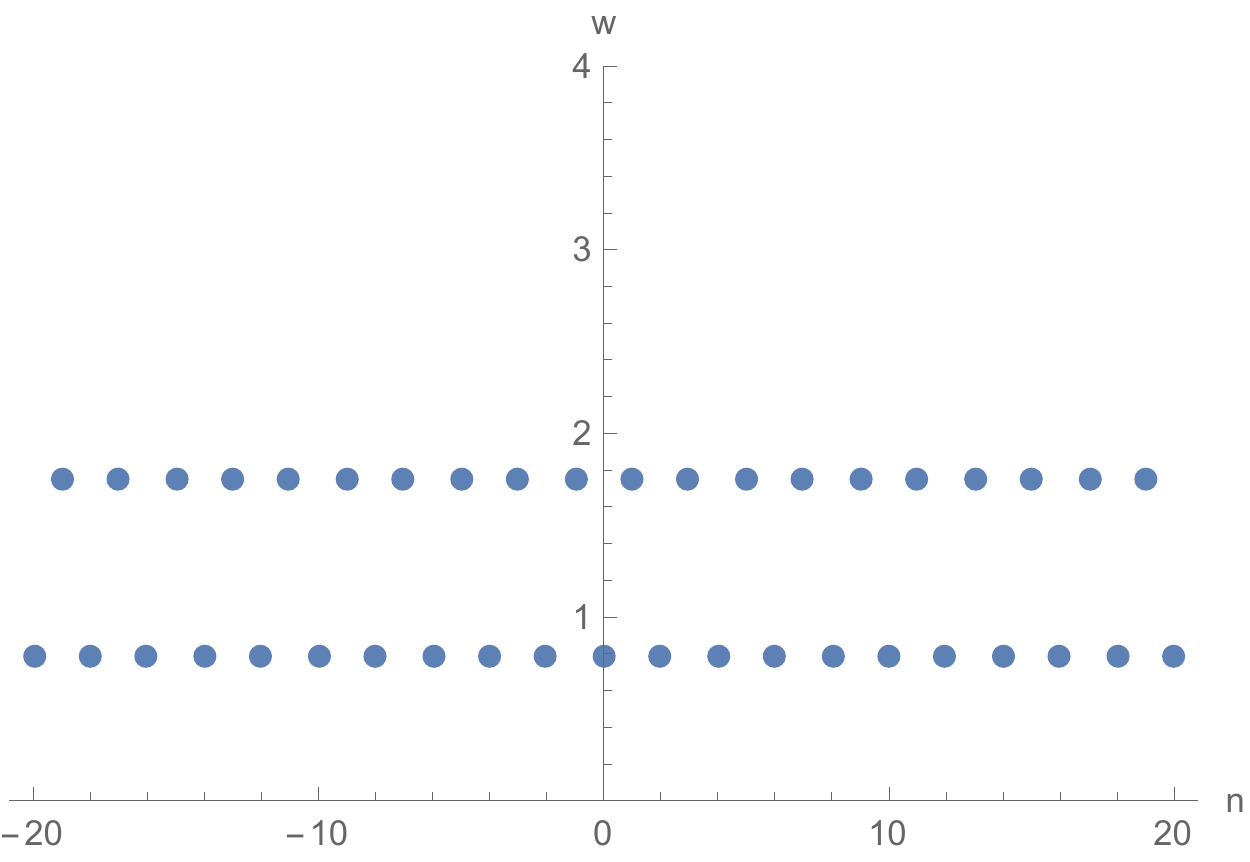}}}
\put(150,-23){\resizebox{!}{3.5cm}{\includegraphics{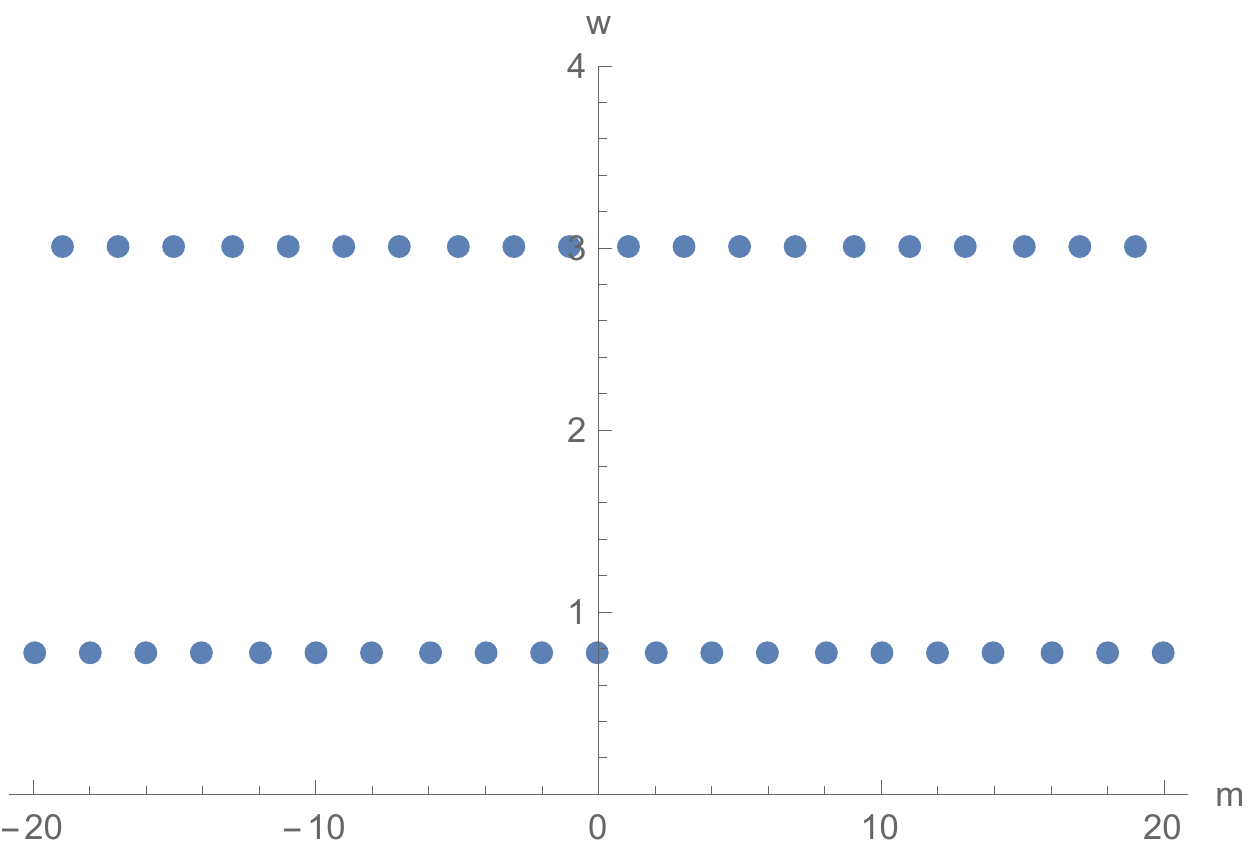}}}
\end{picture}
\end{center}
\vskip 20pt
\begin{center}
\begin{minipage}{15cm}{\footnotesize
\qquad\quad(a)\qquad\qquad\qquad\qquad\qquad\qquad\qquad\qquad\quad (b) \qquad\quad\quad\qquad\qquad\qquad\qquad\qquad\quad\quad (c)\\
{FIGURE 4.} Oscillatory solution $w$ given by \eqref{ABS-1SS} with $p=0.1, q=0.2, k_1=0.5$ and $\rho_1^0c_1=1$:
(a) Shape and motion; (b) Oscillatory solution at $m=2$; (c) Oscillatory solution at $n=2$.}
\end{minipage}
\end{center}

\vskip20pt
\begin{center}
\begin{picture}(120,100)
\put(-180,-23){\resizebox{!}{3.5cm}{\includegraphics{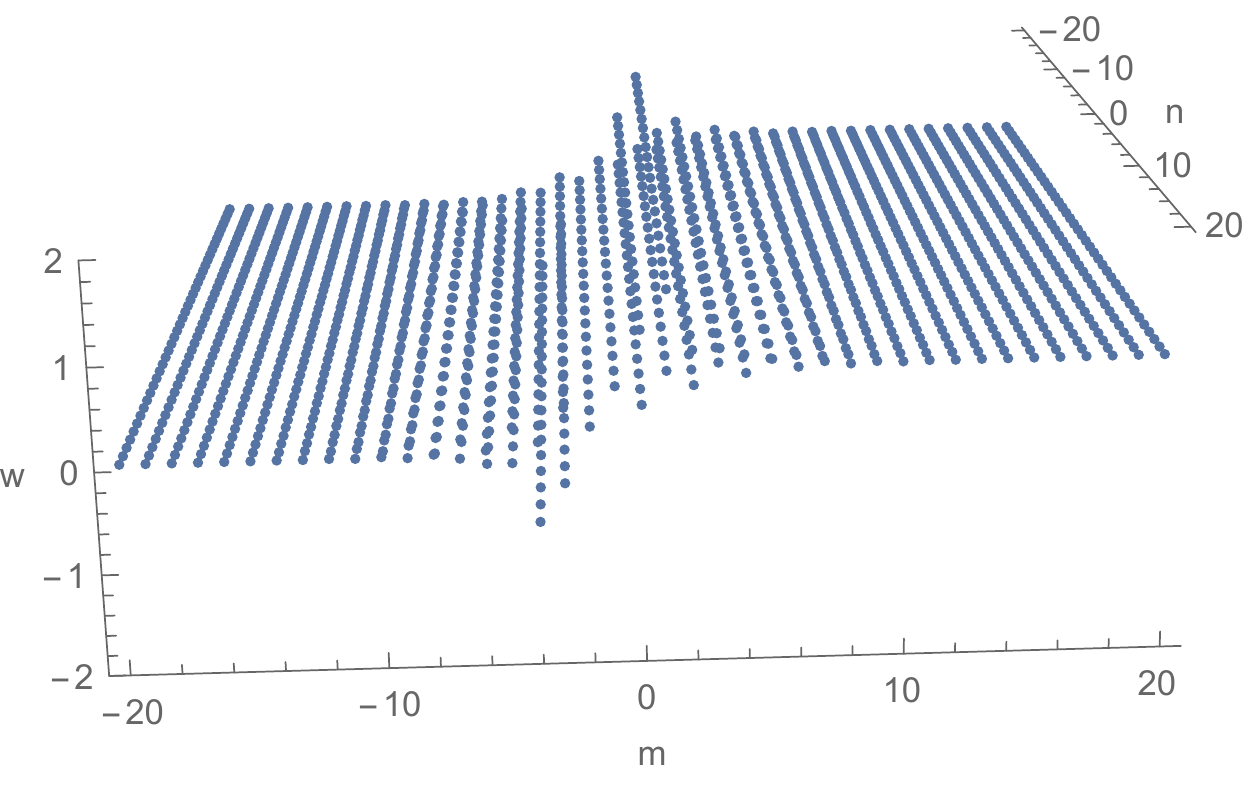}}}
\put(-20,-23){\resizebox{!}{3.5cm}{\includegraphics{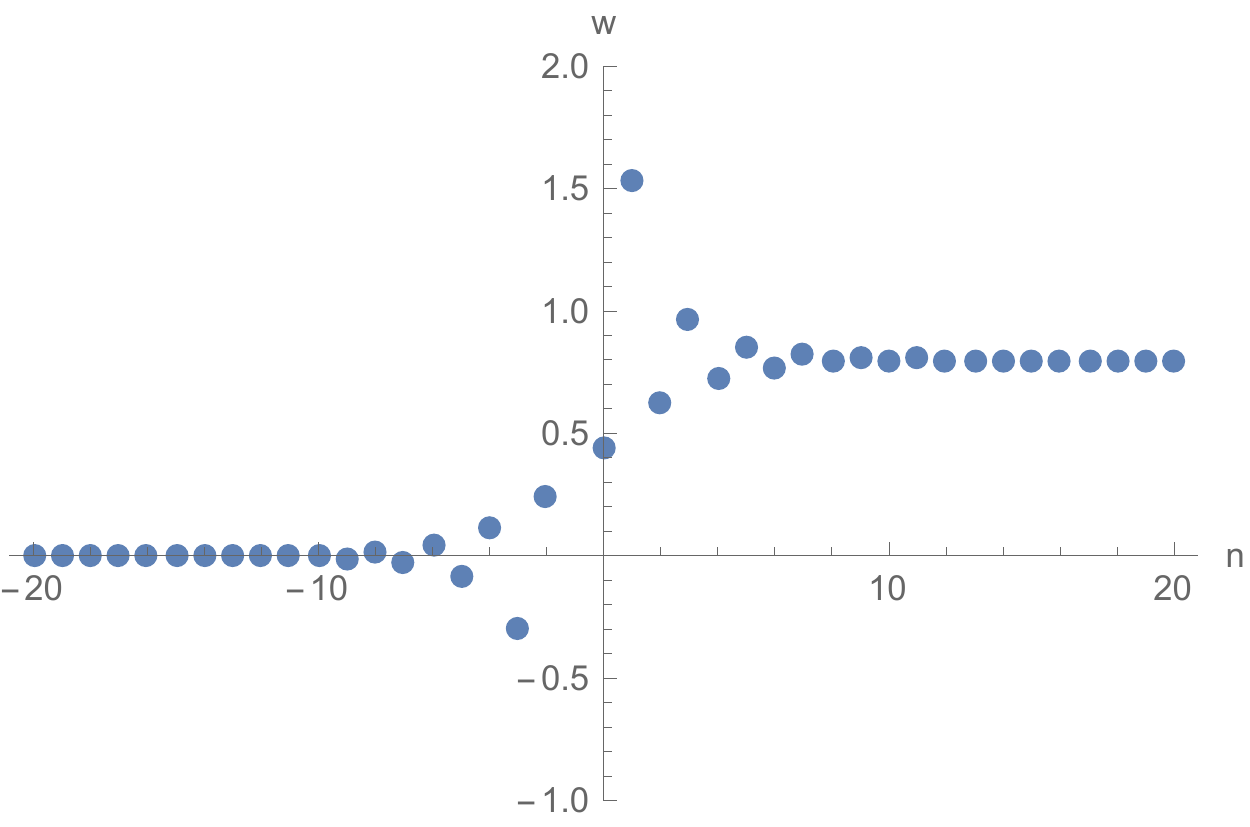}}}
\put(150,-23){\resizebox{!}{3.5cm}{\includegraphics{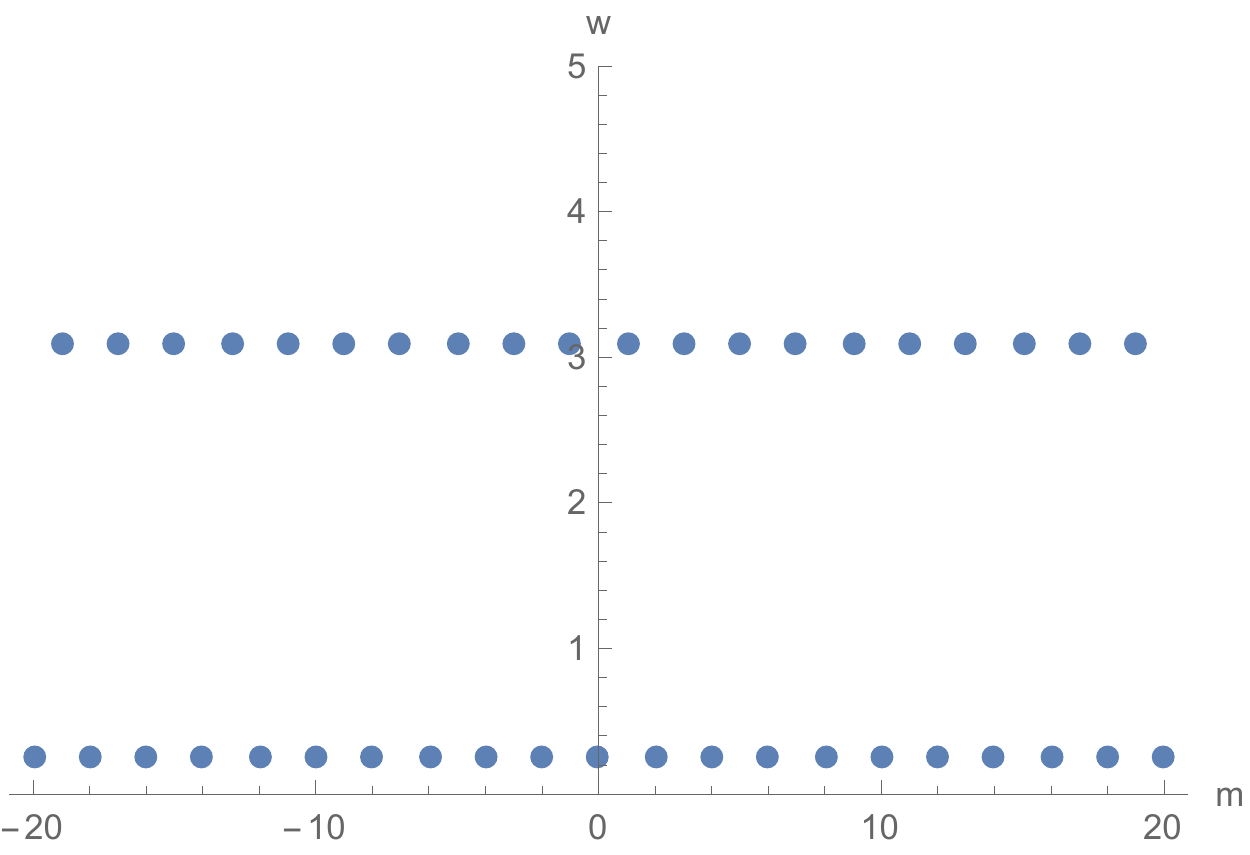}}}
\end{picture}
\end{center}
\vskip 20pt
\begin{center}
\begin{minipage}{15cm}{\footnotesize
\qquad\quad(a)\qquad\qquad\qquad\qquad\qquad\qquad\qquad\qquad\quad (b) \qquad\quad\quad\qquad\qquad\qquad\qquad\qquad\quad\quad (c)\\
{FIGURE 5.} Semi-oscillatory solution $w$ given by \eqref{ABS-1SS} with $p=0.1, q=0.2, k_1=0.15$ and $\rho_1^0c_1=1$:
(a) Shape and motion; (b) Semi-oscillatory solution at $m=2$; (c) Semi-oscillatory solution at $n=-2$.}
\end{minipage}
\end{center}

\section{Lattice BSQ-type equations and solutions}

In this section, we deal with the Cauchy matrix approach for solving the lattice equations of the BSQ-type. Analogous to the
ABS case, here we start from the ``fake'' nonautonomous plane wave factors \eqref{BSQ-pwf}. When $\mathcal {F}_n=\mathcal {G}_m=1$,
the plane wave factors \eqref{BSQ-pwf} become
\begin{align}
& \varrho_{n,m}=\bigg(\frac{p+k}{p+\omega k}\bigg)^n
\bigg(\frac{q+k}{q+\omega k}\bigg)^m\varrho^{0}, \quad \sigma_{n,m}=\bigg(\frac{p+k}{p+\omega^2 k}\bigg)^n
\bigg(\frac{q+k}{q+\omega^2 k}\bigg)^m\sigma^{0},
\end{align}
firstly proposed in the direct linearizing transform for the lattice Gel'fand-Dikii hierarchy
\cite{GD}.  When $\mathcal {F}_n=\omega^n$ or $\mathcal {G}_m=\omega^m$, since $\omega^3=1$, they are
\begin{subequations}
\begin{align}
& \varrho_{n,m}=\frac{p+\omega^{1-n}k}{p+\omega k}
\bigg(\frac{q+k}{q+\omega k}\bigg)^m\varrho^{0}, \quad \sigma_{n,m}=\frac{p+k}{p+\omega^{-n}k}
\bigg(\frac{q+k}{q+\omega^2 k}\bigg)^m\sigma^{0}, \quad \text{if}~\mathcal {F}_n=\omega^n,~\mathcal {G}_m=1, \\
& \varrho_{n,m}=\bigg(\frac{p+k}{p+\omega k}\bigg)^n\frac{q+\omega^{1-m}k}{q+\omega k}\varrho^{0},
\quad \sigma_{n,m}=\bigg(\frac{p+k}{p+\omega^2 k}\bigg)^n
\frac{q+k}{q+\omega^{-m}k}\sigma^{0},  \quad \text{if}~\mathcal {F}_n=1,~\mathcal {G}_m=\omega^m, \
\end{align}
\begin{align}
& \varrho_{n,m}=\frac{p+\omega^{1-n}k}{p+\omega k}\frac{q+\omega^{1-m}k}{q+\omega k}\varrho^{0},
\quad \sigma_{n,m}=\frac{p+k}{p+\omega^{-n}k}
\frac{q+k}{q+\omega^{-m}k}\sigma^{0}, \quad \text{if}~\mathcal {F}_n=\omega^n,~\mathcal {G}_m=\omega^m.
\end{align}
\end{subequations}

Results collected in this section are based on Ref. \cite{FZZ-JNMP}.

\subsection{DES and master functions}

Let us consider the following DES
\begin{subequations}
\label{BSQ-DES}
\begin{align}
\label{BSQ-SE}
& \bL\bM+\bM\bL'=\br\ts, \\
\label{eq:ce-2a}
& \wt{\br}=(\mathcal {F}_np\bI+\bL)\br, \quad \wh{\br}=(\mathcal {G}_mq\bI+\bL)\br, \\
\label{eq:ce-2b}
& \wt{\ts}=\ts(\mathcal {F}_np\bI-\bL')^{-1}, \quad \wh{\ts}=\ts(\mathcal {G}_mq\bI-\bL')^{-1},
\end{align}
\end{subequations}
where functions $\mathcal {F}_n,~\mathcal {G}_m$ satisfy
$\mathcal {F}_n^3=\mathcal {G}_m^3=1$; $p,~q\in \mathbb{C}$ are lattice parameters;
$\bL=\text{Diag}(\bL_1,\ \bL_2)$, $\bL'=\text{Diag}(-\omega\bL_1,\ -\omega^2\bL_2)$ are known constant matrices
with $\omega^2+\omega+1=0$ while $\bM,\ \br$ and $\ts$ depend on $(n,m)$. We suppose that
$\bL_i\in \mathbb{C}_{N_i\times N_i},~(i=1,2)$ satisfy $\mathcal{E}(\bL_1)\bigcap \mathcal{E}(\omega\bL_1)=\varnothing$,
$\mathcal{E}(\bL_2)\bigcap \mathcal{E}(\omega^2\bL_2)=\varnothing$ and $\mathcal{E}(\bL_1)\bigcap \mathcal{E}(\omega^2\bL_2)=\varnothing$,
where $N_1+N_2=N$.

To proceed, we introduce the following similarity invariant quantities involving the matrix $\bM$:
\begin{align}
\label{BSQ-Wijab}
W^{(i,j)}(a,b)=\ts (b\bI+\bL')^j(\bI+\bM)^{-1}(a\bI+\bL)^i \br ,
\end{align}
for $i,~j\in \mathbb{Z}$ and $a,~b\in \mathbb{C}$. 
For the shift relations of the master function $W^{(i,j)}(a,b)$, we have the following result.
\begin{Proposition}
\label{ABS-Prop-Sij-sf}
For the master function $W^{(i,j)}(a,b)$ defined by \eqref{BSQ-Wijab} with $\bM, \bL, \bL', \br,\ts$ satisfying DES \eqref{BSQ-DES},
the following relations
\begin{subequations}
\label{BSQ-Wij-dyna}
\begin{align}
\label{BSQ-Wij-dyna-1}
&  (\mathcal {F}_np+b)\wt{W}^{(i,j)}(a,b)-\wt{W}^{(i,j+1)}(a,b)=(\mathcal {F}_np-a)W^{(i,j)}(a,b)+W^{(i+1,j)}(a,b)\nn \\
& \qquad\qquad\qquad\qquad\qquad\qquad\qquad\qquad\quad -\wt{W}^{(i,0)}(a,b)W^{(0,j)}(a,b),\\
\label{BSQ-Wij-dyna-2}
&  \bigg[\prod_{h=1}^{2}\left(\mathcal {F}_np+\omega^h(b-E_1)\right)W^{(i,j)}(a,b)\bigg]=
\bigg[\prod_{h=1}^{2}\left(\mathcal {F}_np+\omega^h(E_2-a)\right)\wt{W}^{(i,j)}(a,b)\bigg] \nn \\
& ~-\sum_{l=1}^{2}\omega^l\bigg[\prod_{h=2}^l\left(\mathcal {F}_np+\omega^{h-1}(b-E_1)\right)W^{(i,0)}(a,b)\bigg]
\cdot \bigg[\prod_{h=l+1}^{2}\left(\mathcal {F}_np+\omega^h(E_2-a)\right)\wt{W}^{(0,j)}(a,b)\bigg],
\end{align}
\end{subequations}
 involving the shift $\wt{\phantom{a}}$ hold.  Operators $E_1,~E_2$ in \eqref{BSQ-Wij-dyna} are defined by their actions on the indices $i$ and $j$: $E_1
W^{(i,j)}(a,b)=W^{(i,j+1)}(a,b)$ and $E_2 W^{(i,j)}(a,b)=W^{(i+1,j)}(a,b)$.
These equations also hold when $(p,\wt{\phantom{a}},\mathcal {F}_n)$ is replaced by $(q,\wh{\phantom{a}},\mathcal {G}_m)$.
\end{Proposition}

\begin{proof}
Let us first pay attention to the Sylvester equation \eqref{BSQ-SE}, i.e.,
\begin{align}
\label{BSQ-SE-1}
\bL\bM+\bM\bL'=\br\ts.
\end{align}
Subtracting $(\mathcal {F}_np\bI+\bL)^{-1}*$\eqref{BSQ-SE-1} from $\wt{\eqref{BSQ-SE-1}}*(\mathcal {F}_np\bI-\bL')^{-1}$, then we have
\begin{align}
\label{BSQ-M-p1}
\wt{\bM}(\mathcal {F}_np\bI-\bL')=(\mathcal {F}_np\bI+\bL)\bM,
\end{align}
which gives rise to
\begin{align}
\label{BSQ-M-p2}
\wt{\bM}-\bM=\br\wt{\ts},
\end{align}
in the light of the $\wt{\phantom{a}}$-shift of the Sylvester equation \eqref{BSQ-SE}. We now introduce an
auxiliary vector function
\begin{align}
\label{BSQ-via}
\bv^{(i)}(a)=(\bI+\bM)^{-1}(a\bI+\bL)^i \br.
\end{align}
Applying \eqref{BSQ-M-p2} on the vector function $\bv^{(i)}(a)$ and through direct calculation, we find
\begin{subequations}
\label{BSQ-vib-p}
\begin{align}
\label{BSQ-vib-pa}
& \wt{\bv}^{(i)}(a)=(\mathcal {F}_np-a)\bv^{(i)}(a)+\bv^{(i+1)}(a)-\wt{W}^{(i,0)}(a,b)\bv^{(0)}(a), \\
\label{BSQ-vib-pb}
& \bigg[\prod_{h=1}^{3}(\omega^h\mathcal {F}_np\bI-\bL')\bv^{(i)}(a)\bigg]=
\bigg[\prod_{h=1}^{2}(\mathcal {F}_np+\omega^h(E_3-a))\wt{\bv}^{(i)}(a)\bigg] \nn \\
&  -\sum_{l=1}^{2}\omega^l\bigg[\prod_{h=2}^l(\mathcal {F}_np+\omega^{h-1}(b-E_1))W^{(i,0)}(a,b)\bigg]\cdot\bigg[\prod_{h=l+1}^{2}
(\mathcal {F}_np+\omega^h(E_3-a))\wt{\bv}^{(0)}(a)\bigg],
\end{align}
\end{subequations}
where operator $E_3$ is defined by $E_3 \bv^{(i)}(a)=\bv^{(i+1)}(a)$. Furthermore, multiplying \eqref{BSQ-vib-p} from the left by the row
vector $\ts(b\bI+\bL')^j$, we arrive at the shift relations for master functions $W^{(i,j)}(a,b)$.
\end{proof}

\subsection{Lattice BSQ-type equations}

With the $\wt{\phantom{a}}$-shift relations \eqref{BSQ-vib-p} as well as their $\wh{\phantom{a}}$-$q$ counterparts in hand,
we can construct several lattice BSQ-type equations. To do this, we introduce the objects:
\begin{subequations}
\label{BSQ-fun-def}
\begin{align}
& v_a := W^{(-1,0)}(a,0)-1, \quad w_b := W^{(0,-1)}(0,b)-1, \\
& s_a := W^{(-1,1)}(a,0)-a, \quad t_b := W^{(1,-1)}(0,b)-b, \\
& r_a := W^{(-1,2)}(a,0)-a^2, \quad z_b :=W^{(2,-1)}(0,b)-b^2,
\end{align}
and
\begin{align}
\label{eq:Wij-def}
W^{(i,j)}=W^{(i,j)}(0,0), \quad \eta=W^{(0,0)}, \quad s_{a,b}=W^{(-1,-1)}(a,b)-(a+b)^{-1},
\end{align}
\end{subequations}
with $b\neq -a$. For the sake of describing the lattice equations, we denote $G(a,b)=a^3+b^3$.
The equations are presented in the following manner. (For the detailed calculation, one can refer to \cite{FZZ-JNMP}.)

\noindent{\it 1. lattice equations with $\eta,~ v_a$ and $s_a$}:
\begin{subequations}
\label{BSQ-eq-xivsa}
\begin{align}
& \wt{s}_a=(\mathcal {F}_np+\eta)\wt{v}_a-(\mathcal {F}_np-a)v_a, \quad \wh{s}_a=(\mathcal {G}_mq+\eta)\wh{v}_a-(\mathcal {G}_mq-a)v_a, \\
& (\mathcal {F}_np+\mathcal {G}_mq-\wh{\wt{\eta}}+s_a/v_a)(\mathcal {F}_np-\mathcal {G}_mq+\wh{\eta}-\wt{\eta})
=(p^{(1)}_a\wt{v}_a-q^{(1)}_a\wh{v}_a)/v_a,
\end{align}
\end{subequations}
where $p^{(1)}_a=G(p,-a)/(\mathcal {F}_np-a)$ and $q^{(1)}_a=G(q,-a)/(\mathcal {G}_mq-a)$.
Under the point transformation
\begin{align}
v_a=x/x_a, \quad s_a=(y-v_ay_a)/x_a, \quad \eta=z-z_0,
\end{align}
with
\begin{subequations}
\begin{align}
& x_a=\prod^{n-1}_{l=0}(\mathcal {F}_lp-a)^{-1}\prod^{m-1}_{h=0}(\mathcal {G}_hq-a)^{-1}c_1,\\
& y_a=x_a(z_0-c_2),\\
& z_0=c_2(n+m)-\sum^{n-1}_{l=0}\mathcal {F}_lp-\sum^{m-1}_{h=0}\mathcal {G}_hq+c_3,
\end{align}
\end{subequations}
we obtain a closed-form lattice equation
\begin{subequations}
\label{BSQ-A2-1}
\begin{align}
& \wt{y}=\wt{x}z-x, \quad \wh{y}=\wh{x}z-x, \\
& y=x(\wh{\wt{z}}-3c_2)+(G(p,-a)\wt{x}-G(q,-a)\wh{x})/(\wh{z}-\wt{z}),
\end{align}
\end{subequations}
where and whereafter $\{c_i\}$ are constants. Equation \eqref{BSQ-A2-1} is the (A-2) equation formulated in \cite{H-BSQ}.

\noindent{\it 2. lattice equations with $\eta,~ w_b$ and $t_b$}:
\begin{subequations}
\label{BSQ-eq-xiwtb}
\begin{align}
& t_b=(\mathcal {F}_np+b)\wt{w}_b-(\mathcal {F}_np-\wt{\eta})w_b, \quad
  t_b=(\mathcal {G}_mq+b)\wh{w}_b-(\mathcal {G}_mq-\wh{\eta})w_b, \\
& (\mathcal {F}_np+\mathcal {G}_mq+\eta-\wh{\wt{t}}_b/\wh{\wt{w}}_b)(\mathcal {F}_np-\mathcal {G}_mq+\wh{\eta}-\wt{\eta})
=(p^{(2)}_b\wh{w}_b-q^{(2)}_b\wt{w}_b)/\wh{\wt{w}}_b,
\end{align}
\end{subequations}
where $p^{(2)}_b=p^{(1)}_a|_{a\rightarrow -b}$ and $q^{(2)}_b=q^{(1)}_a|_{a\rightarrow -b}$.
Under the point transformation
\begin{align}
w_b=x/x_b, \quad t_b=(y-w_by_b)/x_b, \quad \eta=z-z_0,
\end{align}
with
\begin{subequations}
\begin{align}
& x_b=\prod^{n-1}_{l=0}(-\mathcal {F}_lp-b)\prod^{m-1}_{h=0}(-\mathcal {G}_hq-b) c_1,\\
& y_b=x_b(z_0-c_2), \\
& z_0=-c_2(n+m)-\sum^{n-1}_{l=0}\mathcal {F}_l
p-\sum^{m-1}_{h=0}\mathcal {G}_hq+c_3,
\end{align}
\end{subequations}
we obtain another closed-form lattice equation
\begin{subequations}
\label{BSQ-A2-2}
\begin{align}
& y=x\wt{z}-\wt{x}, \quad y=x\wh{z}-\wh{x}, \\
& \wh{\wt{y}}=\wh{\wt{x}}(z-3c_2)+(G(p,b)\wh{x}-G(q,b)\wt{x})/(\wh{z}-\wt{z}),
\end{align}
\end{subequations}
which is the antisymmetric version of \eqref{BSQ-A2-1} with
\begin{align}
p\rightarrow -p, \quad q\rightarrow -q, \quad n\rightarrow -n, \quad m\rightarrow -m, \quad a\rightarrow b.
\end{align}

\noindent{\it 3. lattice equations with $\eta,~W^{(1,0)}$ and $W^{(0,1)}$}:
\begin{subequations}
\label{BSQ-eq-xiW1001}
\begin{align}
& \mathcal {F}_np \wt{\eta}-\wt{W}^{(0,1)}=\mathcal {F}_np\eta+W^{(1,0)}-\eta\wt{\eta}, \quad
\mathcal {G}_mq \wh{\eta}-\wh{W}^{(0,1)}=\mathcal {G}_mq\eta+W^{(1,0)}-\eta\wh{\eta}, \\
& \wh{\wt{W}}^{(1,0)}+W^{(0,1)}=\mathcal {F}_n\mathcal {G}_mpq-(\mathcal {F}_np+\mathcal {G}_mq-\wh{\wt{\eta}})(\mathcal {F}_np+\mathcal {G}_mq+\eta)\nn \\
& \qquad\qquad\qquad\qquad  +G(p,-q)/(\mathcal {F}_np-\mathcal {G}_mq+\wh{\eta}-\wt{\eta}).
\end{align}
\end{subequations}
This equation can be rewritten as the standard lpBSQ equation
\begin{subequations}
\label{lpBSQ}
\begin{align}
& \wt{z}=x\wt{x}-y, \quad \wh{z}=x\wh{x}-y, \\
& z=x\wh{\wt{x}}-\wh{\wt{y}}+G(p,-q)/(\wh{x}-\wt{x}),
\end{align}
\end{subequations}
under the point transformations
\begin{align}
\eta=x+x_0, \quad W^{(1,0)}=y+x_0\eta-y_0, \quad W^{(0,1)}=z+x_0\eta-z_0,
\end{align}
with
\begin{subequations}
\label{lpBSQ-xyz0}
\begin{align}
& x_0=\sum^{n-1}_{l=0}\mathcal {F}_lp+\sum^{m-1}_{h=0}\mathcal {G}_hq+c_1,\\
& y_0=x_0^2/2+\big(\sum^{n-1}_{l=0}\mathcal {F}_l^2p^2+\sum^{m-1}_{h=0}\mathcal {G}_h^2q^2+c_2\big)/2+c_3,\\
& z_0=x_0^2/2-\big(\sum^{n-1}_{l=0}\mathcal {F}_l^2p^2+\sum^{m-1}_{h=0}\mathcal {G}_h^2q^2+c_2\big)/2-c_3.
\end{align}
\end{subequations}
\noindent{\bf Remark 4.} \textit{In terms of the choices of $\mathcal {F}_l$ and
$\mathcal {G}_h$, one gets various seed solutions for the lpBSQ equation \eqref{lpBSQ}. We list these solutions as follows. \\
When $\mathcal {F}_l=\mathcal {G}_h=1$, the seed solution is
\begin{subequations}
\label{BSQ-lss}
\begin{align}
& x_0=np+mq+c_1, \\
& y_0=x_0^2/2+(np^2+mq^2+c_2)/2+c_3, \\
& z_0=x_0^2/2-(np^2+mq^2+c_2)/2-c_3.
\end{align}
\end{subequations}
When $\mathcal {F}_l=1$ and $\mathcal {G}_h=\omega^h$, the seed solution reads
\begin{subequations}
\label{BSQ-oss-1}
\begin{align}
& x_0=np-\omega^mq/(1-\omega)+c_1, \\
& y_0=x_0^2/2+\big(n p^2-\omega^{2m}q^2/(1-\omega^2)+c_2\big)/2+c_3, \\
& z_0=x_0^2/2-\big(n p^2-\omega^{2m}q^2/(1-\omega^2)+c_2\big)/2-c_3.
\end{align}
\end{subequations}
Similarly, we can have seed solutions as $\mathcal {F}_l=\omega^l$ and $\mathcal {G}_h=1$.
When $\mathcal {F}_l=\omega^l$ and $\mathcal {G}_h=\omega^h$, we get
\begin{subequations}
\label{BSQ-oss-2}
\begin{align}
& x_0=-\omega^np/(1-\omega)-\omega^mq/(1-\omega)+c_1, \\
& y_0=x_0^2/2+\big(-\omega^{2n}p^2/(1-\omega^2)-\omega^{2m}q^2/(1-\omega^2)+c_2\big)/2+c_3, \\
& z_0=x_0^2/2-\big(-\omega^{2n}p^2/(1-\omega^2)-\omega^{2m}q^2/(1-\omega^2)+c_2\big)/2-c_3.
\end{align}
\end{subequations}
Solution \eqref{BSQ-lss} has been reported in \cite{HZ-BSQ}, while \eqref{BSQ-oss-1} and \eqref{BSQ-oss-2} are new, which can be understood as
semi-oscillatory and oscillatory seed solutions for the lpBSQ equation \eqref{lpBSQ}, respectively.}

\noindent{\it 4. lattice equations with $v_a,~w_b$ and $s_{a,b}$}: The first system reads
\begin{subequations}
\label{BSQ-eq-vawbsab-1}
\begin{align}
& (\mathcal {F}_np-a)s_{a,b}-(\mathcal {F}_np+b)\wt{s}_{a,b}=\wt{v}_a w_b, \quad
(\mathcal {G}_mq-a)s_{a,b}-(\mathcal {G}_mq+b)\wh{s}_{a,b}=\wh{v}_a w_b, \\
&  v_{a}\wh{\wt{w}}_{b}=w_b\frac{\frac{p^{(1)}_a}{\mathcal {F}_np+b}\wt{v}_a\wh{w}_b-\frac{q^{(1)}_a}
{\mathcal {G}_mq+b}\wh{v}_a\wt{w}_b}
{(\mathcal {F}_np+b)\wt{w}_b-(\mathcal {G}_mq+b)\wh{w}_b}-\frac{G(a,b)}{(\mathcal {F}_np+b)(\mathcal {G}_mq+b)}s_{a,b},
\end{align}
\end{subequations}
and the second one is
\begin{subequations}
\label{BSQ-eq-vawbsab-2}
\begin{align}
& (\mathcal {F}_np-a)s_{a,b}-(\mathcal {F}_np+b)\wt{s}_{a,b}=\wt{v}_a w_b, \quad
(\mathcal {G}_mq-a)s_{a,b}-(\mathcal {G}_mq+b)\wh{s}_{a,b}=\wh{v}_a w_b, \\
& v_{a}\wh{\wt{w}}_{b}=w_b\frac{\frac{p^{(2)}_b}{\mathcal {F}_np-a}\wt{v}_a\wh{w}_b-\frac{q^{(2)}_b}
{\mathcal {G}_mq-a}\wh{v}_a\wt{w}_b}
{(\mathcal {F}_np+b)\wt{w}_b-(\mathcal {G}_mq+b)\wh{w}_b}-\frac{G(a,b)}{(\mathcal {F}_np-a)(\mathcal {G}_mq-a)}\wh{\wt{s}}_{a,b}.
\end{align}
\end{subequations}
Starting from equation \eqref{BSQ-eq-vawbsab-1}
or \eqref{BSQ-eq-vawbsab-2}, by the same point transformation
\begin{subequations}
\label{BSQ-svw}
\begin{align}
& s_{a,b}=\prod^{n-1}_{l=0}\left(\frac{\mathcal {F}_lp-a}{\mathcal {F}_lp+b}\right)
\prod^{m-1}_{h=0}\left(\frac{\mathcal {G}_hq-a}{\mathcal {G}_hq+b}\right)x,\\
& v_a=\prod^{n-1}_{l=0}(\mathcal {F}_lp-a)\prod^{m-1}_{h=0}(\mathcal {G}_hq-a)y,\\
& w_b=\prod^{n-1}_{l=0}(\mathcal {F}_lp+b)^{-1}\prod^{m-1}_{h=0}(\mathcal {G}_hq+b)^{-1}z,
\end{align}
\end{subequations}
we arrive at
\begin{subequations}
\label{BSQ-mxyz-1}
\begin{align}
& \wt{y}z=x-\wt{x}, \quad \wh{y}z=x-\wh{x}, \\
& y\wh{\wt{z}}=z(G(p,-a)\wt{y}\wh{z}-G(q,-a)\wh{y}\wt{z})/(\wt{z}-\wh{z})-G(a,b)x,
\end{align}
\end{subequations}
and
\begin{subequations}
\label{BSQ-mxyz-2}
\begin{align}
& \wt{y}z=x-\wt{x}, \quad \wh{y}z=x-\wh{x},\\
& y\wh{\wt{z}}=z(G(p,b)\wt{y}\wh{z}-G(q,b)\wh{y}\wt{z})/(\wt{z}-\wh{z})-G(a,b)\wh{\wt{x}},
\end{align}
\end{subequations}
respectively. Equation \eqref{BSQ-mxyz-1} is the (C-3) equation given in \cite{HZ-BSQ}
and \eqref{BSQ-mxyz-2} is the the reversal symmetry version of \eqref{BSQ-mxyz-1} with
\begin{align}
n\rightarrow -n, \quad m\rightarrow -m, \quad y\rightarrow z, \quad z\rightarrow -y, \quad a\rightarrow -b.
\end{align}

The observation of \eqref{BSQ-mxyz-1} and \eqref{BSQ-mxyz-2} sharing the same solution
\eqref{BSQ-svw} gives rise to
\begin{subequations}
\label{BSQ-mxyz-12}
\begin{align}
& \wt{y}z=x-\wt{x}, \quad \wh{y}z=x-\wh{x}, \\
& y\wh{\wt{z}}=z(P_{a,b}\wt{y}\wh{z}-Q_{a,b}\wh{y}\wt{z})/(\wt{z}-\wh{z})-G_{a,b}(x+\wh{\wt{x}}),
\end{align}
\end{subequations}
where
\begin{align}
P_{a,b}=(G(p,-a)+G(p,b))/2, \quad Q_{a,b}=(G(q,-a)+G(q,b))/2, \quad \Omega_{a,b}=G(a,b)/2.
\end{align}
By the transformation
\begin{align}
x=(x_1-\Omega_{a,b})/(2\Omega_{a,b}(x_1+\Omega_{a,b})), \quad y=y_1/(x_1+\Omega_{a,b}), \quad z=z_1/(x_1+\Omega_{a,b}),
\end{align}
from \eqref{BSQ-mxyz-12} we arrive at the (C-4) equation
\begin{subequations}
\begin{align}
& \wt{y}_1z_1=x_1-\wt{x}_1, \quad \wh{y}_1z_1=x_1-\wh{x}_1,\\
& y_1\wh{\wt{z}}_1=z_1(P_{a,b}\wh{z}_1\wt{y}_1-Q_{a,b}\wt{z}_1\wh{y}_1)/(\wt{z}_1-\wh{z}_1)-x_1\wh{\wt{x}}_1+\Omega_{a,b}^2.
\end{align}
\end{subequations}

\subsection{Exact solutions}

Analogous to the ABS case, since master function $W^{(i,j)}(a,b)$ is similarity invariant, the construction of solutions for the lattice BSQ-type equation can be achieved by solving the canonical DES
\begin{subequations}
\label{BSQ-DES-ca}
\begin{align}
\label{BSQ-SE-ca}
& \La\bM+\bM\La'=\br\ts, \\
& \wt{\br}=(\mathcal {F}_np\bI+\La)\br, \quad \wh{\br}=(\mathcal {G}_mq\bI+\La)\br, \\
& \wt{\ts}=\ts(\mathcal {F}_np\bI-\La')^{-1}, \quad \wh{\ts}= \ts(\mathcal {G}_mq\bI-\La')^{-1},
\end{align}
\end{subequations}
where $\La=\text{Diag}(\La_1,\La_2)$ and $\La'=\text{Diag}(-\omega\La_1,-\omega^2\La_2)$
are Jordan canonical forms of the matrices $\bL$ and $\bL'$. To guarantee the solvability of
\eqref{BSQ-SE-ca}, we suppose $\mathcal{E}(\La_1)\bigcap \mathcal{E}(\omega\La_1)=\varnothing$,
$\mathcal{E}(\La_2)\bigcap \mathcal{E}(\omega^2\La_2)=\varnothing$ and $\mathcal{E}(\La_1)\bigcap \mathcal{E}(\omega^2\La_2)=\varnothing$.

Let $\La_\gamma$ be a general block diagonal matrix
\begin{subequations}
\begin{align}
\label{BSQ-Lai}
& \La_1={\rm Diag}\big(\La^{\tyb{N$_{11}$}}_{\ty{D}}(\{k_{1,\iota}\}_{\iota=1}^{N_{11}}),
~\La^{\tyb{N$_{12}$}}_{\ty{J}}(k_{1,N_{11}+1}),~\cdots,~\La^{\tyb{N$_{1s}$}}_{\ty{J}}(k_{1,N_{11}+(s-1)})\big), \\
& \La_2={\rm Diag}\big(\La^{\tyb{N$_{21}$}}_{\ty{D}}(\{k_{2,\iota}\}_{\iota=1}^{N_{21}}),
~\La^{\tyb{N$_{22}$}}_{\ty{J}}(k_{2,N_{21}+1}),~\cdots,~\La^{\tyb{N$_{2t}$}}_{\ty{J}}(k_{2,N_{21}+(t-1)})\big),
\end{align}
\end{subequations}
where $\sum\limits_{\gamma=1}^{s}N_{1\gamma}=N_1$ and $\sum\limits_{\gamma=1}^{t}N_{2\gamma}=N_2$.
We summarize the most general mixed solution as follows (A set of notations is given in the Appendix C).
\begin{Thm}
For the DES \eqref{BSQ-DES-ca} with generic
\begin{align}
\label{La-gen-T}
\La=\mathrm{Diag}(\La_1,~\La_2), \quad \La'=\mathrm{Diag}(-\omega\La_1,~-\omega^2\La_2),
\end{align}
we have solutions
\begin{subequations}
\label{BSQ-rts-D}\begin{align}
& \br=(\br^{(1)}_1, \br^{(2)}_1, \cdots, \br^{(s)}_1;\br^{(1)}_2, \br^{(2)}_2, \cdots, \br^{(t)}_2)^{\st}, \\
& \ts=(\ts^{(1)}_1, \ts^{(2)}_1, \cdots, \ts^{(s)}_1;\ts^{(1)}_2, \ts^{(2)}_2, \cdots, \ts^{(t)}_2),
\end{align}
\end{subequations}
with
\begin{subequations}
\begin{align}
& \br^{(1)}_i=(r^{(1)}_{i,1}, r^{(1)}_{i,2}, \cdots, r^{(1)}_{i,N_{i1}}), \quad r^{(1)}_{i,\iota}=\tau_{i,\iota},~i=1,2,~\iota=1,2,\ldots,N_{i1}, \\
& \ts^{(1)}_j=(s^{(1)}_{j,1}, s^{(1)}_{j,2}, \cdots, s^{(1)}_{j,N_{j1}}),~~s^{(1)}_{j,\kappa}=\varsigma_{j,\kappa},~~j=1,2,~\kappa=1,2,\ldots,N_{i1}, \\
& \br^{(l)}_1=(r^{(l)}_{1,1}, r^{(l)}_{1,2}, \cdots, r^{(l)}_{1,N_{1l}}),~~
r^{(l)}_{1,\io}=\frac{1}{(\iota-1)!}\partial^{\iota-1}_{k_{1,N_{11}+(l-1)}}\tau_{1,N_{11}+(l-1)},\nn \\
& \qquad \qquad l=2,3,\ldots,s,~\iota=1,2,\ldots,N_{1l}, \\
& \br^{(l)}_2=(r^{(l)}_{2,1}, r^{(l)}_{2,2}, \cdots, r^{(l)}_{2,N_{2l}}),~~
r^{(l)}_{2,\iota}=\frac{1}{(\iota-1)!}\partial^{\iota-1}_{k_{2,N_{21}+(l-1)}}\tau_{2,N_{21}+(l-1)},\nn \\
& \qquad \qquad l=2,3,\ldots,t,~\iota=1,2,\ldots,N_{2l}, \\
& \ts^{(l)}_1=(s^{(l)}_{1,1}, s^{(l)}_{1,2}, \cdots, s^{(l)}_{1,N_{1l}}),~~
s^{(l)}_{1,\iota}=\frac{1}{(N_{1l}-\kappa)!}\partial^{N_{1l}-\kappa}_{k_{1,N_{11}+(l-1)}}\varsigma_{1,N_{11}+(l-1)},\nn \\
&\qquad \qquad l=2,3,\ldots,s,~\kappa=1,2,\ldots,N_{1l}, \\
& \ts^{(l)}_2=(s^{(l)}_{2,1}, s^{(l)}_{2,2}, \cdots, s^{(l)}_{2,N_{2l}}),~~
s^{(l)}_{2,\iota}=\frac{1}{(N_{2l}-\kappa)!}\partial^{N_{2l}-\kappa}_{k_{2,N_{21}+(l-1)}}\varsigma_{2,N_{21}+(l-1)},\nn \\
&\qquad \qquad l=2,3,\ldots,t,~\kappa=1,2,\ldots,N_{2l},
\end{align}
\end{subequations}
and $\bM=\bF\bG\bH$, in which
\begin{subequations}
\begin{align}
&\bF=\mathrm{Diag}\bigl(
\La^{\tyb{N$_{11}$}}_{\ty{D}}(\{r^{(1)}_{1,\io}\}^{N_{11}}_{\io=1}),
\bT^{\tyb{N$_{12}$}}(\{r^{(2)}_{1,\io}\}_{\io=1}^{N_{12}}),\ldots,
\bT^{\tyb{N$_{1s}$}}(\{r^{(s)}_{1,\io}\}_{\io=1}^{N_{1s}}); \nn \\
&\qquad\qquad\quad \La^{\tyb{N$_{21}$}}_{\ty{D}}(\{r^{(1)}_{2,\io}\}^{N_{21}}_{\io=1}),
\bT^{\tyb{N$_{22}$}}(\{r^{(2)}_{2,\io}\}_{\io=1}^{N_{22}}),\ldots,
\bT^{\tyb{N$_{2t}$}}(\{r^{(t)}_{2,\io}\}_{\io=1}^{N_{2t}})
\bigr), \\
&\bH=\mathrm{Diag}\bigl(
\La^{\tyb{N$_{11}$}}_{\ty{D}}(\{s^{(1)}_{1,\kp}\}^{N_{11}}_{\kp=1}),
\bH^{\tyb{N$_{12}$}}(\{s^{(2)}_{1,\kp}\}_{\kp=1}^{N_{12}}),\ldots,
\bH^{\tyb{N$_{1s}$}}(\{s^{(s)}_{1,\kp}\}_{\kp=1}^{N_{1s}}); \nn \\
&\qquad\qquad\quad \La^{\tyb{N$_{21}$}}_{\ty{D}}(\{s^{(1)}_{2,\kp}\}^{N_{21}}_{\kp=1}),
\bH^{\tyb{N$_{22}$}}(\{s^{(2)}_{2,\kp}\}_{\kp=1}^{N_{22}}),\ldots,
\bH^{\tyb{N$_{2t}$}}(\{s^{(t)}_{2,\kp}\}_{\kp=1}^{N_{2t}})
\bigr), \\
& \bG=\left(\begin{array}{cc}
\bG^{(1)} & \bG^{(2)}\\
\bG^{(3)} & \bG^{(4)}\\
\end{array}
\right)_{N\times N},
\end{align}
where
\begin{align}
\begin{array}{ll}
\bG^{(1)}=(\bG^{(1)}_{i,j})_{s\times s},~
\bG^{(2)}=(\bG^{(2)}_{i,j})_{s\times t}, ~
\bG^{(3)}=(\bG^{(3)}_{i,j})_{t\times s}, ~
\bG^{(4)}=(\bG^{(4)}_{i,j})_{t\times t},
\end{array}
\end{align}
\end{subequations}
with
\begin{subequations}
\begin{align}
& \bG^{(1)}_{1,1}=\bG'_{\ty{DD}}(\{k_{1,\io}\}_{\io=1}^{N_{11}};\{\og k_{1,\kp}\}_{\kp=1}^{N_{11}}), \\
& \bG^{(1)}_{i,1}=\bG'_{\ty{JD}}(k_{1,N_{11}+i-1};\{\og k_{1,\kappa}\}_{\kp=1}^{N_{11}}), \quad 2\leq i\leq s-1,\\
& \bG^{(1)}_{1,j}=\bG'_{\ty{DJ}}(\{k_{1,\io}\}_{\io=1}^{N_{11}};\og k_{1,N_{11}+j-1}), \quad 2\leq j\leq s-1,\\
& \bG^{(1)}_{i,j}=\bG'_{\ty{JJ}}(k_{1,N_{11}+i-1};\og k_{1,N_{11}+j-1}), \quad 2\leq i,j\leq s-1,
\end{align}
\end{subequations}
and
\begin{subequations}
\begin{align}
& \bG^{(2)}_{1,1}=\bG'_{\ty{DD}}(\{k_{1,\io}\}_{\io=1}^{N_{11}};\{\og^2k_{2,\kp}\}_{\kp=1}^{N_{21}}), \\
& \bG^{(2)}_{i,1}=\bG'_{\ty{JD}}(k_{1,N_{11}+i-1};\{\og^2k_{2,\kappa}\}_{\kp=1}^{N_{21}}), \quad 2\leq i\leq s-1,\\
& \bG^{(2)}_{1,j}=\bG'_{\ty{DJ}}(\{k_{1,\io}\}_{\io=1}^{N_{11}};\og^2k_{2,N_{21}+j-1}), \quad 2\leq j\leq t-1,\\
& \bG^{(2)}_{i,j}=\bG'_{\ty{JJ}}(k_{1,N_{11}+i-1};\og^2k_{2,N_{21}+j-1}), \quad 2\leq i\leq s-1,~2\leq j\leq t-1,
\end{align}
\end{subequations}
and
\begin{subequations}
\begin{align}
& \bG^{(3)}_{1,1}=\bG'_{\ty{DD}}(\{k_{2,\io}\}_{\io=1}^{N_{21}};\{\og k_{1,\kp}\}_{\kp=1}^{N_{11}}), \\
& \bG^{(3)}_{i,1}=\bG'_{\ty{JD}}(k_{2,N_{21}+i-1};\{\og k_{1,\kappa}\}_{\kp=1}^{N_{11}}), \quad 2\leq i\leq t-1,\\
& \bG^{(3)}_{1,j}=\bG'_{\ty{DJ}}(\{k_{2,\io}\}_{\io=1}^{N_{21}};\og k_{1,N_{11}+j-1}), \quad 2\leq j\leq s-1,\\
& \bG^{(3)}_{i,j}=\bG'_{\ty{JJ}}(k_{2,N_{21}+i-1};\og k_{1,N_{11}+j-1}), \quad 2\leq i\leq t-1,~2\leq j\leq s-1,
\end{align}
\end{subequations}
and
\begin{subequations}
\begin{align}
& \bG^{(4)}_{1,1}=\bG'_{\ty{DD}}(\{k_{2,\io}\}_{\io=1}^{N_{21}};\{\og^2k_{2,\kp}\}_{\kp=1}^{N_{21}}), \\
& \bG^{(4)}_{i,1}=\bG'_{\ty{JD}}(k_{2,N_{21}+i-1};\{\og^2k_{2,\kp}\}_{\kp=1}^{N_{21}}), \quad 2\leq i\leq t-1,\\
& \bG^{(4)}_{1,j}=\bG'_{\ty{DJ}}(\{k_{2,\io}\}_{\io=1}^{N_{21}};\og^2 k_{2,N_{21}+j-1}), \quad 2\leq j\leq t-1,\\
& \bG^{(4)}_{i,j}=\bG'_{\ty{JJ}}(k_{2,N_{21}+i-1};\og^2k_{2,N_{21}+j-1}), \quad 2\leq i,j\leq t-1.
\end{align}
\end{subequations}
\end{Thm}

In this instance, we would utilize the lpBSQ equation denoted by \eqref{lpBSQ} as a case study to present the exact solutions explicitly.
When $N_1=N_2=1$ and $k_{1,1}=k_1$, $k_{2,1}=k_2$, we obtain a solution for the lpBSQ equation \eqref{lpBSQ}
\begin{align}
x=\eta-x_0, \quad y=W^{(1,0)}-x_0\eta+y_0, \quad z= W^{(0,1)}-x_0\eta+z_0,
\end{align}
with \eqref{lpBSQ-xyz0} and
\begin{subequations}
\label{BSQ-1ss}
\begin{align}
& \eta=\frac{H_0(\varrho+\sigma)+A_1\varrho\sigma}{H_0+H_1\varrho+H_2\sigma+H_3\varrho\sigma}, \\
& W^{(1,0)}=\frac{H_0(k_1\varrho+k_2\sigma)+A_2\varrho\sigma}{H_0+H_1\varrho+H_2\sigma+H_3\varrho\sigma}, \\
& W^{(0,1)}=\frac{-\omega H_0(k_1\varrho+\omega k_2\sigma)-\omega A_3\varrho\sigma}{H_0+H_1\varrho+H_2\sigma+H_3\varrho\sigma},
\end{align}
\end{subequations}
where
\begin{subequations}
\begin{align}
& A_1=-(2\omega+1)(k_1^3-k_2^3), \quad  A_2=(k_1+k_2)A_1, \quad A_3=(1-\omega^2)(k_1^3-k_2^3)(\omega^2 k_1+k_2), \\
& H_0=3k_1k_2(2k_1k_2-\omega k_1^2-\omega^2 k_2^2), \quad H_1=(\omega^2-1)k_2(\omega k_1^2-2k_1k_2+\omega^2 k_2^2), \\
& H_2=(\omega-1)k_1(\omega k_1^2-2k_1k_2+\omega^2 k_2^2), \quad H_3=(k_1-k_2)(\omega^2 k_2-\omega k_1), \\
& \varrho=\prod^{n-1}_{l=0}\bigg(\frac{\mathcal {F}_lp+k_1}{\mathcal {F}_lp+\omega k_1}\bigg)
\prod^{m-1}_{h=0}\bigg(\frac{\mathcal {G}_hq+k_1}{\mathcal {G}_hq+\omega k_1}\bigg)\varrho^{0}, \quad
\sigma=\prod^{n-1}_{l=0}\bigg(\frac{\mathcal {F}_lp+k_2}{\mathcal {F}_lp+\omega^2 k_2}\bigg)
\prod^{m-1}_{h=0}\bigg(\frac{\mathcal {G}_hq+k_2}{\mathcal {G}_hq+\omega^2 k_2}\bigg)\sigma^{0}.
\end{align}
\end{subequations}
If $k_1=k_2$, then \eqref{BSQ-1ss} is simplified to
\begin{subequations}
\begin{align}
\label{BSQ-1ss-k12}
\eta=\frac{H_0(\varrho+\sigma_1)}{H_0+H_1\varrho+H_2\sigma_1},
\quad W^{(1,0)}=k_1\eta,
\quad W^{(0,1)}=\frac{-\omega H_0k_1(\varrho+\omega \sigma_1)}{H_0+H_1\varrho+H_2\sigma_1},
\end{align}
with
\begin{align}
H_0=3k_1, \quad H_1=1-\omega^2, \quad H_2=1-\omega, \quad \sigma_1=\sigma|_{k_2\rightarrow k_1}.
\end{align}
\end{subequations}
When $\mathcal {F}_l=\mathcal {G}_h=1$, solution \eqref{BSQ-1ss-k12} is exactly the soliton solution (3.10) given in \cite{FZZ-JNMP}, which is
real if we take $\sigma^{0}=\varrho^{0^*}$.
When $\mathcal {F}_l=\omega^l$ or $\mathcal {G}_h=\omega^h$, \eqref{BSQ-1ss-k12} 
provides (semi-)oscillatory solution for the lpBSQ equation \eqref{lpBSQ}. Since $\sigma$ in this case can not be viewed as
complex conjugate of $\varrho$, (semi-)oscillatory solution \eqref{BSQ-1ss-k12} is complex.

\section{Conclusions}
\label{sec:6}

This paper presents a generalized Cauchy matrix scheme, based on the previous work in \cite{Nijhoff-ABS,ZZ-SAM-2013,FZZ-JNMP}, which can be used to reconstruct solutions for all ABS equations (except for $\mathrm{Q4}$) and some lattice BSQ-type equations. Starting from the DES \eqref{ABS-DES} involving $f_n$ and $g_m$ where $f^2_n=g^2_m=1$, we define master equations
$S^{(i,j)},~S(a,b),~V(a)$, which possess several properties, such as symmetric property, similarity invariance and shift relations.
Benifiting from these shift relations, some lattice KdV-type equations were derived,
including lpKdV \eqref{lpKdV-w}, lpmKdV \eqref{eq:v}, lSKdV \eqref{lSKdV} and NQC \eqref{NQC}
equations. Among these equations, the first three are autonomous since they can
be transformed into autonomous equations through simple point transformations (see also \cite{FZ-OS}).
Although the NQC equation \eqref{NQC} is nonautonomous as $f_n=(-1)^n$ or $g_m=(-1)^m$, it can still be used to construct
solutions for the $\mathrm{Q3}_{\delta}$ equation. Furthermore, solutions
for other equations in the list \eqref{ABS} are derived by using the degeneration
relation depicted in Figure 2. The most general mixed solutions to the Jordan canonical DES \eqref{ABS-DES-JC} are presented,  based on which
exact solutions to the lattice equations in the  ABS list except for $\mathrm{Q4}$ are derived. In terms of
\begin{align*}
(f_n,g_m)=(1,1)~ \text{or} ~ ((-1)^n,(-1)^m)~ \text{or}~(1,(-1)^m) ~ \text{or}~((-1)^n,1),
\end{align*}
we construct soliton, oscillatory and semi-oscillatory solutions, respectively. Different from soliton solution,
the oscillatory solution has periodic property (see solutions \eqref{ABS-1SS}-\eqref{ABS-JBS}).
To obtain the solutions of lattice BSQ-type equations, we adopt a similar strategy, in other words,
we introduce two auxiliary functions $\mathcal {F}_n$ and $\mathcal {G}_m$ with $\mathcal {F}^3_n=\mathcal {G}^3_m=1$
in the DES \eqref{BSQ-DES}. Then we derive some three-component lattice BSQ-type equations as a closed-form from the shift relations
\eqref{BSQ-Wij-dyna} as well as their hat-$q$ counterpart, which correspond to (A-2), (B-2),  (C-3) and (C-4) equations
firstly introduced in \cite{H-BSQ}. Soliton, oscillatory and semi-oscillatory solutions are presented by setting
\begin{align*}
(\mathcal {F}_n,\mathcal {G}_m)=(1,1), ~ (\omega^n,\omega^m)~\text{and}~ (1,\omega^m) ~ \text{or}~(\omega^n,1).
\end{align*}
When $(\mathcal {F}_n,\mathcal {G}_m)=(\omega^n,\omega^m)$, the resulting solution still has periodic property (see \eqref{BSQ-1ss}).
With regard to the solution \eqref{BSQ-1ss-k12}, it appears real in the case of soliton and $\sigma^{0}=\varrho^{0^*}$. While in the case of
oscillatory or semi-oscillatory, this solution is complex regardless of whether $\sigma^{0}=\varrho^{0^*}$ or not.

In conclusion, we would like to highlight some important modifications made to the Cauchy matrix framework \cite{ZZ-SAM-2013, FZZ-JNMP}. First of all, constant lattice parameters $p$ and $q$ are replaced by $\mathbb{F}_np$ and $\mathbb{G}_mq$, where in the lattice ABS case $(\mathbb{F}_n,\mathbb{G}_m)=(f_n,g_m)$
and in the lattice BSQ case $(\mathbb{F}_n,\mathbb{G}_m)=(\mathcal {F}_n,\mathcal {G}_m)$, respectively. Although these changes in principle enable the generation of nonautonomous lattice equations  (cf. \cite{WZZ-CAMC,FZ-ROMP}), we show that the
constraints $f^2_n=g_m^2=1$ and $\mathcal{F}^3_n=\mathcal {G}_m^3=1$ enable one to derive the autonomous lattice ABS/BSQ equations.
This modification leads to the emergence of (semi-)oscillatory solutions in the resulting lattice equations. As, in the continuous case, oscillatory factors $((-1)^n,(-1)^m)$ or $(\omega^n,\omega^m)$
break differentiability and do not appear in analytic solutions, there is no continuum limit for the oscillatory solutions. However, the semi-oscillatory plane wave factor allows for the straight continuum limit on the discrete
exponential part. Thus, the semi-discrete KdV-type equations \cite{MZ-TMP} and semi-discrete BSQ-type equations \cite{HZ-DBSQ} are still capable of producing semi-oscillatory solutions.
Furthermore, this scheme can be generalized to the entire lattice Gel'fand-Dikii hierarchy \cite{GD} (also see \cite{TZZ-SIGMA}) by initializing the DES \eqref{BSQ-DES} with $\bL'=\text{Diag}(-\omega\bL_1,\ -\omega^2\bL_2,\ \ldots, \ -\omega^{N-1}\bL_{N-1})$ and $\mathcal {F}^N_n=\mathcal {G}^N_m=1$, where $\omega^N=1$. However, it should be noted that this scheme cannot be applied to the extended lattice BSQ-type equations \cite{ZZF-SAPM} since it is not possible to introduce independent variables $n$ or $m$ into the three solutions $\omega_1(k),~\omega_2(k)$ and $\omega_3(k)=k$ of the third order polynomial equation of symmetric form
\begin{align*}
G_3(\omega,k)=\sum\limits_{j=1}^3\alpha_j(\omega^j-k^j)=0, \quad \alpha_3\equiv 1,
\end{align*}
with coefficients $\{\alpha_j\}$ and parameter $k$.
Finally, we construct soliton solutions of the nonautonomous ABS lattice equations with the help of the bilinear method \cite{SZZ}.  Thus we can set $p_n=f_np$ and $q_m=g_mq$, and regain the oscillatory solutions of $\mathrm{H1}$, $\mathrm{H2}$, $\mathrm{H3}_{\delta}$ and $\mathrm{Q1}_{\delta}$ equations from the perspective of bilinear structure. This operation can be naturally generalized to the BSQ case \cite{HZ-DBSQ}, which will be done in the future.

\vskip 20pt
\section*{Acknowledgments}
This project is supported by the National Natural Science Foundation of China
(Nos. 12071432, 12001369), the Natural Science Foundation of Zhejiang Province (No. LY17A010024)
and Shanghai Sailing Program (No. 20YF1433000).

\vskip 20pt
\section*{Data Availibility Statement}
Data sharing not applicable to this article as no datasets were generated or
analyzed during the current study.

\vskip 20pt
\section*{Conflict of interest}
There are no conflicts of interest to declare.

\section*{Appendix}
\appendix

\setcounter{equation}{0}
\renewcommand\theequation{A.\arabic{equation}}

\section{Some lattice equations in the ABS list}\label{sec:LoN}

After reparameterizing the lattice parameters in the ABS list, some of the lattice equations can be described as
\begin{subequations}
\label{ABS}
\begin{align}
\label{Q3}
\mathrm{Q3}_{\delta}:~~ & P(u\wh{u}+\wt{u}\wh{\wt{u}})-Q(u\wt{u}+\wh{u}\wh{\wt{u}})
=(p^2-q^2)\bigg(\wt{u}\wh{u}+u\wh{\wt{u}}+\frac{\delta^2}{4PQ}\bigg), \\
\label{Q2}
\mathrm{Q2}:~~& (q^2-a^2)(u-\wh{u})(\wt{u}-\wh{\wt{u}})-(p^2-a^2)(u-\wt{u})(\wh{u}-\wh{\wt{u}}) \nn \\
& +(p^2-a^2)(q^2-a^2)(q^2-p^2)(u+\wt{u}+\wh{u}+\wh{\wt{u}}) \nn \\
& =(p^2-a^2)(q^2-a^2)(q^2-p^2)\big((p^2-a^2)^2+(q^2-a^2)^2-(p^2-a^2)(q^2-a^2)\big), \\
\label{Q1}
\mathrm{Q1}_{\delta}:~~ & (q^2-a^2)(u-\wh{u})(\wt{u}-\wh{\wt{u}})-(p^2-a^2)(u-\wt{u})(\wh{u}-\wh{\wt{u}})
=\frac{\delta^2 a^4(p^2-q^2)}{(p^2-a^2)(q^2-a^2)}, \\
\label{H3}
\mathrm{H3}_{\delta}:~~ & P(a^2-q^2)(u\wt{u}+\wh{u}\wh{\wt{u}})-
Q(a^2-p^2)(u\wh{u}+\wt{u}\wh{\wt{u}})=\delta(p^2-q^2), \\
\label{H2}
\mathrm{H2}:~~ & (u-\wh{\wt{u}})(\wt{u}-\wh{u})+(p^2-q^2)
(u+\wt{u}+\wh{u}+\wh{\wt{u}})=p^4-q^4,\\
\label{H1}
\mathrm{H1}:~~ & (u-\wh{\wt{u}})(\wh{u}-\wt{u})=p^2-q^2,
\end{align}
\end{subequations}
where $\delta$ is a constant and in \eqref{Q3} $(p,P)=\mathfrak{p}$ and $(q,Q)=\mathfrak{q}$ are the points on the elliptic curve
\begin{align}
\label{ell-Q3}
\{(x,X)|X^2=(x^2-a^2)(x^2-b^2)\},
\end{align}
and in \eqref{H3}
\begin{align}
\label{eq:parcurves}
P^2=a^2-p^2, \quad Q^2=a^2-q^2.
\end{align}
The $\mathrm{A1}_\delta$ and $\mathrm{A2}$ equations are omitted here since $\mathrm{A1}_\delta$
and $\mathrm{Q1}_\delta$ are equivalent under a point transformation,
so do $\mathrm{A2}$ and $\mathrm{Q3}_{\delta=0}$.

\section{List of notations for solutions to ABS list}
\label{sec:LoN}

\setcounter{equation}{0}
\renewcommand\theequation{B.\arabic{equation}}

We introduce some notations where the subscripts $_D$ and $_J$ usually correspond to
the cases of $\Ga$ being diagonal and being of Jordan-block, respectively.
\begin{subequations}
\label{ABS-nota}
\begin{align}
& \hbox{plane wave factor:}~~\rho_s=\prod^{n-1}_{i=0}\bigg(\frac{f_ip+k_s}{f_ip-k_s}\bigg)
\prod^{m-1}_{j=0}\bigg(\frac{g_jq+k_s}{g_jq-k_s}\bigg)\rho_s^{0},\\
& N\mathrm{\hbox{-}th~order~vector:}~~\br_{\hbox{\tiny{\it D}}}^{\hbox{\tiny{[{\it N}]}}}(\{k_s\}_{1}^{N})=(\rho_1, \rho_2, \cdots, \rho_N)^{\st},\\
& N\mathrm{\hbox{-}th~order~vector:}~~\br_{\ty{J}}^{\tyb{N}}(k_1)=\Bigl(\rho_1, \frac{\partial_{k_1}\rho_1}{1!},
\cdots, \frac{\partial^{N-1}_{k_1}\rho_1}{(N-1)!}\Bigr)^{\st},\\
& N\times N ~\mathrm{matrix:}~~\Ga^{\tyb{N}}_{\ty{D}}(\{k_s\}^{N}_{1})=\mathrm{Diag}(k_1, k_2, \cdots, k_N),\\
& N\times N ~\mathrm{matrix:}~~\Ga^{\tyb{N}}_{\ty{J}}(a)
=\left(\begin{array}{cccccc}
a & 0    & 0   & \cdots & 0   & 0 \\
1   & a  & 0   & \cdots & 0   & 0 \\
0   & 1  & a   & \cdots & 0   & 0 \\
\vdots &\vdots &\vdots &\vdots &\vdots &\vdots \\
0   & 0    & 0   & \cdots & 1   & a
\end{array}\right), \\
& N\times N ~\mathrm{matrix:}~~\bF^{\tyb{N}}_{\ty{D}}(\{k_s\}^{N}_{1})=\mathrm{Diag}(\rho_1, \rho_2, \cdots, \rho_N),\\
& N\times N ~\mathrm{matrix:}~~\bH^{\tyb{N}}_{\ty{D}}(\{c_s\}^{N}_{1})=\mathrm{Diag}(c_1, c_2, \cdots, c_N),\\
& N\times N ~\mathrm{matrix:}~~\bF^{\tyb{N}}_{\ty{J}}(k_1)
=\left(
\begin{array}{ccccc}
\rho_1 & 0 & 0 & \cdots & 0\\
\frac{\partial_{k_1}\rho_1}{1!} & \rho_1 & 0 & \cdots & 0\\
\frac{\partial^{2}_{k_1}\rho_1}{2!} &\frac{\partial_{k_1}\rho_1}{1!} & \rho_1 & \cdots & 0\\
\vdots &\vdots &\vdots & \ddots & \vdots\\
\frac{\partial^{N-1}_{k_1}\rho_1}{(N-1)!} & \frac{\partial^{N-2}_{k_1}\rho_1 }{(N-2)!} & \frac{\partial^{N-3}_{k_1}\rho_1}{(N-3)!} & \cdots & \rho_1
\end{array}
\right),\\
& N\times N ~\mathrm{matrix:}~~\bH^{\tyb{N}}_{\ty{J}}(\{c_s\}^{N}_{1})
=\left(\begin{array}{ccccc}
c_1 & \cdots  & c_{N-2}  & c_{N-1} & c_N\\
c_2 & \cdots & c_{N-1}  & c_N & 0\\
c_3 &\cdots & c_N & 0 & 0\\
\vdots &\vdots & \vdots & \vdots & \vdots\\
c_N & \cdots & 0 & 0 & 0
\end{array}
\right),
\end{align}
\begin{align}
& N\times N ~\mathrm{matrix:}~~\bG^{\tyb{N}}_{\ty{D}}(\{k_s\}^{N}_{1})
=(g_{i,j})_{N\times N},~~~g_{i,j}=\frac{1}{k_i+k_j},\\
& N_1\times N_2 ~\mathrm{matrix:}~~\bG^{\tyb{N$_1$,N$_2$}}_{\ty{DJ}}(\{k_s\}^{N_1}_{1};a)
=(g_{i,j})_{N_1\times N_2},~~~g_{i,j}=-\Bigl(\frac{-1}{k_i+a}\Bigr)^j,\\
& N_1\times N_2 ~\mathrm{matrix:}~~\bG^{\tyb{N$_1$,N$_2$}}_{\ty{JJ}}(a;b)
=(g_{i,j})_{N_1\times N_2},~~~g_{i,j}=\mathrm{C}^{i-1}_{i+j-2}\frac{(-1)^{i+j}}{(a+b)^{i+j-1}},\\
& N\times N ~\mathrm{matrix:}~~\bG^{\tyb{N}}_{\ty{J}}(a)=\bG^{\tyb{N,N}}_{\ty{JJ}}(a;a)
=(g_{i,j})_{N\times N},~~~g_{i,j}=\mathrm{C}^{i-1}_{i+j-2}\frac{(-1)^{i+j}}{(2a)^{i+j-1}},\label{G-J-a}
\end{align}
\end{subequations}
where
\[\mathrm{C}^{i}_{j}=\frac{j!}{i!(j-i)!},~~(j\geq i).\]

The $N$-th order matrix in the following form
\begin{align}
\mathcal{B}=\left(\begin{array}{cccccc}
a_0 & 0    & 0   & \cdots & 0   & 0 \\
a_1 & a_0  & 0   & \cdots & 0   & 0 \\
a_2 & a_1  & a_0 & \cdots & 0   & 0 \\
\vdots &\vdots &\cdots &\vdots &\vdots &\vdots \\
a_{N-1} & a_{N-2} & a_{N-3}  & \cdots &  a_1   & a_0
\end{array}\right)_{N\times N}
\label{A}
\end{align}
with scalar elements $\{a_j\}$ is  a $N$th-order lower triangular Toeplitz matrix.
All such matrices compose a commutative set $\widetilde{G}^{\tyb{N}}$ with respect to matrix multiplication
and the subset
\[G^{\tyb{N}}=\big \{\mathcal{B} \big |~\big. \mathcal{B}\in \widetilde{G}^{\tyb{N}},~|\mathcal{B}|\neq 0 \big\}\]
is an Abelian group.
Such kind of matrices play useful roles in the expression of exact solution for soliton equations \cite{ZDJ-Wron,ZZSZ}.

\section{List of notations for solutions to lattice BSQ-type equations}
\label{sec:LoN}

\setcounter{equation}{0}
\renewcommand\theequation{C.\arabic{equation}}

\begin{itemize}
\item{$N_{i1} \times N_{i1}$ diagonal matrix:
\begin{align}
\La^{[N_{i1}]}_{\ty{D}}(\{k_{i,j}\}_{j=1}^{N_{i1}})={\rm Diag}(k_{i,1},k_{i,2},\ldots,k_{i,N_{i1}}),
\end{align}
}
\item{$N_{ij} \times N_{ij}$ Jordan-block matrix:
\begin{align}
& \La^{[N_{ij}]}_{\ty{J}}(a)
=\left(\begin{array}{cccccc}
a & 0    & 0   & \cdots & 0   & 0 \\
1   & a  & 0   & \cdots & 0   & 0 \\
0   & 1  & a   & \cdots & 0   & 0 \\
\vdots &\vdots &\vdots &\vdots &\vdots &\vdots \\
0   & 0    & 0   & \cdots & 1   & a
\end{array}\right)_{N_{ij} \times N_{ij}},
\end{align}
}
\item{Lower triangular Toeplitz matrices:
\begin{align}
\bT^{\tyb{N}}(\{a_i\}^{N}_{1})
=\left(\begin{array}{cccccc}
a_1 & 0    & 0   & \cdots & 0   & 0 \\
a_2 & a_1  & 0   & \cdots & 0   & 0 \\
a_3 & a_2  & a_1 & \cdots & 0   & 0 \\
\vdots &\vdots &\cdots &\vdots &\vdots &\vdots \\
a_{N} & a_{N-1} & a_{N-2}  & \cdots &  a_2   & a_1
\end{array}\right)_{N\times N},
\label{T}
\end{align}
}
\item{Skew triangular Toeplitz matrix:
\begin{align}
\bH^{\tyb{N}}(\{b_j\}^{N}_{1})
=\left(\begin{array}{ccccc}
b_1 & \cdots  & b_{N-2}  & b_{N-1} & b_{N}\\
b_2 & \cdots & b_{N-1}  & b_{N} & 0\\
b_3 &\cdots & b_{N} & 0 & 0\\
\vdots &\vdots & \vdots & \vdots & \vdots\\
b_{N} & \cdots & 0 & 0 & 0
\end{array}
\right)_{N\times N}.
\label{H}
\end{align}}
\end{itemize}
Meanwhile, the following expressions need to be considered:
\begin{subequations}
\label{BSQ-notations}
\begin{align}
\label{ro-i}
& \hbox{plane wave factor:}~~\tau_{i,\iota}=\prod_{l=0}^{n-1}(\mathcal {F}_lp+k_{i,\iota}) \prod_{h=0}^{m-1}(\mathcal {G}_hq+k_{i,\iota})
\tau_{i,\iota}^0, \quad \text{with~constants}~\tau_{i,\iota}^0,  \\
\label{sg-j}
& \hbox{plane wave factor:}~~\varsigma_{j,\kappa}=\prod_{l=n_0}^{n-1}(\mathcal {F}_lp+\omega^jk_{j,\kappa})^{-1}
\prod_{h=0}^{m-1}(\mathcal {G}_hq+\omega^jk_{j,\kappa})^{-1} \varsigma_{j,\kappa}^0, \nn \\
& \qquad\qquad\qquad\qquad \text{with~constants} ~\varsigma_{j,\kappa}^0, \\
& N_i\times N_j ~\mathrm{matrix:}~~\bG'_{\ty{DD}}(\{k_{i,\iota}\}_{\iota=1}^{N_{i1}};\{\omega^jk_{j,\kappa}\}_{\kappa=1}^{N_{j1}})
=\bigg(\frac{1}{k_{i,\iota}-\omega^jk_{j,\kappa}}\bigg)_{\iota,\kappa}, \\
& N_i\times N_j ~\mathrm{matrix:}~~\bG'_{\ty{DJ}}(\{k_{i,\iota}\}_{\iota=1}^{N_i};\omega^jb)
=\bigg(\frac{1}{(\kappa-1)!}\partial^{\kappa-1}_b\frac{1}{k_{i,\iota}-\omega^jb}\bigg)_{\iota,\kappa}, \\
& N_i\times N_j ~\mathrm{matrix:}~~\bG'_{\ty{JD}}(a;\{\omega^jk_{j,\kappa}\}_{\kappa=1}^{N_j})
=\bigg(\frac{1}{(a-\omega^jk_{j,\kappa})^{\iota}}\bigg)_{\iota,\kappa},\\
& N_i\times N_j ~\mathrm{matrix:}~~\bG'_{\ty{JJ}}(a;\omega^jb)
=\Bigl(\frac{1}{(\kappa-1)!}\partial^{\kappa-1}_b\frac{1}{(a-\omega^jb)^{\iota}}\Bigl)_{\iota,\kappa}.
\end{align}
\end{subequations}


{\small
\begin{thebibliography}{99}

\bibitem{Adler-1998} Adler, V.E.:
        B\"{a}cklund transformation for the Krichever-Novikov equation.
        Int. Math. Res. Not. {\bf 1}, 1--4 (1998)
\bibitem{ABS-2003} Adler, V.E., Bobenko, A.I., Suris, Y.B.:
        Classification of integrable equations on quad-graphs, the consistency approach.
        Commun. Math. Phys. {\bf 233}, 513--543 (2003)
\bibitem{ABS-2009} Adler, V.E., Bobenko, A.I., Suris, Y.B.:
        Discrete nonlinear hyperbolic equations. Classification of integrable cases.
        Funct. Anal. Appl. {\bf 43}, 3--21 (2009)
\bibitem{LKP} Adler, V.E., Bobenko, A.I., Suris, Y.B.:
        Classification of integrable discrete equations of octahedron type.
        Int. Math. Res. Notices {\bf 2012}(8), 1822--1889 (2012)
\bibitem{Atki-BT} Atkinson, J.:
        B\"{a}cklund transformations for integrable lattice equations.
        J. Phys. A: Math. Theor. {\bf 41}(8pp), 135202 (2008)
\bibitem{Bhatia} Bhatia, R., Rosenthal, P.:
         How and why to solve the operator equation $AX-XB=Y$.
         Bull. London Math. Soc. {\bf 29}, 1--21 (1997)
\bibitem{BS-2002} Bobenko, A.I., Suris, Y.B.:
        Integrable systems on quad-graphs.
        Int. Math. Res. Notices {\bf 11}, 573--611 (2002)
\bibitem{Boll} Boll, R.:
        Classification of 3D consistent quad-equations.
        J. Nonlinear Math. Phys. {\bf 18}, 337--365 (2011)
\bibitem{BHQK-FCM} Bridgman, T., Hereman, W., Quispel, G.R.W., van der Kamp, P.H.:
        Symbolic computation of Lax pairs of partial difference equations using consistency around the cube.
        Found. Comput. Math. {\bf 13}, 517--544 (2013)
\bibitem{ABS-IST} Butler, S.:
        Multidimensional inverse scattering of integrable lattice equations.
        Nonlinearity {\bf 25}, 1613--1634 (2012)
\bibitem{H1-IST} Butler, S., Joshi, N.:
        An inverse scattering transform for the lattice potential KdV equation.
        Inverse Prob. {\bf 26}(28pp), 115012 (2010)
\bibitem{FZ-ROMP} Feng, W., Zhao, S.L.:
        Solutions and three-dimensional consistency for nonautonomous extended lattice Boussinesq-type equations.
        Rep. Math. Phys. {\bf 78}(2), 219--243 (2016)
\bibitem{FZ-OS} Feng, W., Zhao, S.L.:
        Oscillatory solutions for lattice Korteweg-de Vries-type equations.
        Z. Naturforsch. A {\bf 73}(2), 91--98 (2018)
\bibitem{FZZ-JNMP} Feng, W., Zhao, S.L., Zhang, D.J.:
        Exact solutions to lattice Boussinesq-type equations.
        J. Nonlinear Math. Phys. {\bf 19}(15pp), 1250031 (2012)
\bibitem{H-2005} Hietarinta, J.:
        Searching for CAC-maps.
        J. Nonlinear Math. Phys. {\bf 12}, 223--230 (2005)
\bibitem{H-BSQ} Hietarinta, J.:
         Boussinesq-like multi-component lattice equations and multi-dimensional consistency.
         J. Phys. A: Math. Theor. {\bf 44}(22pp), 165204 (2011)
\bibitem{HZ-ABS} Hietarinta, J., Zhang, D.J.:
        Soliton solutions for ABS lattice equations: II. Casoratians and bilinearization.
        J. Phys. A: Math. Theor. {\bf 42}(30pp), 404006 (2009)
\bibitem{HZ-BSQ} Hietarinta, J., Zhang, D.J.:
        Multisoliton solutions to the lattice Boussinesq equation.
        J. Math. Phys. {\bf 51}(12pp), 033505 (2010)
\bibitem{HZ-eBSQ} Hietarinta, J., Zhang, D.J.:
        Soliton taxonomy for a modification of the lattice Boussinesq equation.
        SIGMA {\bf 7}(14pp), 061 (2011)
\bibitem{HZ-DBSQ} Hietarinta, J., Zhang, D.J.:
        Discrete Boussinesq-type equations.
        in Nonlinear Systems and their Remarkable Mathematical Structures. Vol. 3, CRC Press, Boca Raton, FL, 54--101 (2022)
\bibitem{MZ-TMP} Mesfun, M., Zhao, S.L.:
         Cauchy matrix scheme for semidiscrete lattice Korteweg-de Vries-type equations.
         Theor. Math. Phys. {\bf 211}(1), 483--497 (2022)
\bibitem{Schwar-1} Nijhoff, F.W.:
        On some ``Schwarzian equations'' and their discrete analogues,
        Eds. A.S. Fokas and I.M. Gel'fand, in: Algebraic Aspects of Integrable Systems:
        In memory of Irene Dorfman, Birkh\"{a}user Verlag, 237--260 (1996)
\bibitem{Nij-Adler} Nijhoff, F.W.:
        Lax pair for the Adler (lattice Krichever-Novikov) system.
        Phys. Lett. A. {\bf 297}, 49--58 (2002)
\bibitem{Nijhoff-ABS} Nijhoff, F.W., Atkinson, J., Hietarinta, J.:
        Soliton solutions for ABS lattice equations: I. Cauchy matrix approach.
        J. Phys. A: Math. Theor. {\bf 42}(34pp), 404005 (2009)
\bibitem{lKdV} Nijhoff, F.W., Capel, H.W.:
        The discrete Korteweg-de Vries equation.
        Acta Appl. Math. {\bf 39}, 133--158 (1995)
\bibitem{GD} Nijhoff, F.W., Papageorgiou, V.G., Capel, H.W., Quispel, G.R.W.:
        The lattice Gel'fand-Dikii hierarchy.
        Inverse Probl. {\bf 8}, 597--621 (1992)
\bibitem{NQC-1983} Nijhoff, F.W., Quispel, G.R.W., Capel, H.W.:
        Direct linearization of nonlinear difference-difference equations.
        Phys. Lett. A. {\bf 97}, 125--128 (1983)
\bibitem{Nijhoff-MDC} Nijhoff, F.W., Walker, A.J.:
        The discrete and continuous Painlev\'{e} VI hierarchy and the Garnier systems.
        Glasgow Math. J. {\bf 43A}, 109--123 (2001)
\bibitem{H1-DT} Shi, Y., Nimmo, J.J.C., Zhang, D.J.:
        Darboux and binary Darboux transformations for discrete integrable systems I. Discrete potential KdV equation.
        J. Phys. A: Math. Theor. {\bf 47}(11pp), 025205 (2014)
\bibitem{H3-DT} Shi, Y., Nimmo, J.J.C., Zhao, J.X.:
        Darboux and binary Darboux transformations for discrete integrable systems. II. Discrete potential mKdV equation.
        SIGMA {\bf 13}(18pp) 036 (2017)
\bibitem{SZZ} Shi, Y., Zhang, D.J., Zhao, S.L.:
        Solutions to the non-autonomous ABS lattice equations: Casoratians and bilinearization (in Chinese).
        Sci. Sin. Math. {\bf 44}(1), 37--54 (2014), arXiv:1201.6478
\bibitem{mBSQ-DT} Shi, Y., Zhao, J.X.:
        Discrete modified Boussinesq equation and Darboux transformation.
        Appl. Math. Lett. {\bf 86} 141--148 (2018)
\bibitem{BT-BSQ} Sun, Y.Y., Sun, W.Y.:
        An update of a B\"{a}cklund transformation and its applications to the Boussinesq system.
        Appl. Math. Comput. {\bf 421}(14pp), 126964 (2022)
\bibitem{Syl} Sylvester, J.:
		Sur l'equation en matrices $px=xq$.
		C. R. Acad. Sci. Paris {\bf 99}, 67--71, 115--116 (1884)
\bibitem{TZZ-SIGMA} Tela, G.Y., Zhao, S.L., Zhang, D.J.:
        On the fourth-order lattice Gel'fand-Dikii equations.
        SIGMA {\bf 19}(30pp), 007 (2023)
\bibitem{Tongas} Tongas, A.,  Nijhoff, F.W.:
        The Boussinesq integrable system: Compatible lattice and continuum structures.
        Glasgow Math. J. {\bf 47(A)}, 205--219 (2005)
\bibitem{WE-1973} Wahlquist, H.D., Estabrook, F.B.:
        B\"{a}cklund transformation for solutions of the Korteweg-de Vries equation.
        Phys. Rev. Lett. {\bf 31}, 1386--1390 (1973)
\bibitem{Walker} Walker, A.J.:
        Similarity reductions and integrable lattice equations,
        PhD thesis, Leeds University (2001)
\bibitem{WZZ-CAMC} Wang, X., Zhang, D.J., Zhao, S.L.:
        Solutions to the non-autonomous ABS lattice equations: generalized Cauchy matrix approach.
        Commun. Appl. Math. Comput. {\bf 32}(1), 1--25 (2018)
\bibitem{OS-lpKdV} Wu, H., Zheng, H.C., Zhang, D.J.:
        Oscillatory solutions of lattice potential Korteweg-de Vries equation.
        Commun. Appl. Math. Comput. {\bf 30}(4), 482--489 (2016)
\bibitem{ZKZ-CAC} Zhang, D.D., van der Kamp, P.H., Zhang, D.J.:
        Multi-component extension of CAC systems.
        SIGMA {\bf 16}(30pp), 060 (2020)
\bibitem{ZDJ-Wron} Zhang, D.J.:
		Notes on solutions in Wronskian form to soliton equations: KdV-type.
		arXiv:nlin.SI/0603008, (2006).
\bibitem{HZ-H1} Zhang, D.J., Hietarinta, J.:
        Generalized solutions for the H1 model in ABS List of lattice equations,
        Nonl. Mod. Math. Phys: Proceedings of the First International Workshop, AIP Conference Proceedings,
        {\bf 1212}, 154--161 (2010)
\bibitem{ZZ-SAM-2013} Zhang, D.J., Zhao, S.L.:
        Solutions to ABS lattice equations via generalized Cauchy matrix approach.
        Stud. Appl. Math. {\bf 131}, 72--103 (2013)
\bibitem{ZZF-SAPM} Zhang, D.J., Zhao, S.L., Nijhoff, F.W.:
        Direct linearization of extended lattice BSQ systems.
        Stud. Appl. Math. {\bf 129}, 220--248 (2012)
\bibitem{ZZSZ} Zhang, D.J., Zhao, S.L., Sun, Y.Y., Zhou, J.:
		Solutions to the modified Korteweg-de Vries equation (review).
		Rev. Math. Phys. {\bf 26}(42pp), 14300064 (2014)

\end{thebibliography}}
\end{document}